\documentclass[11pt,a4paper]{article}
\pdfoutput=1

\usepackage[T1]{fontenc}
\usepackage[utf8]{inputenc}
\usepackage[a4paper,margin=1in]{geometry}
\usepackage{lmodern}
\usepackage{graphicx}
\usepackage{amsmath,amssymb,amsfonts}
\usepackage{amsthm}
\usepackage{xcolor}
\usepackage{booktabs}
\usepackage{array}
\usepackage{braket}
\usepackage{tikz}
\usetikzlibrary{arrows.meta,decorations.pathmorphing,decorations.pathreplacing,calc}
\usepackage{float}
\usepackage{enumitem}
\usepackage{microtype}
\usepackage[hidelinks]{hyperref}

\usepackage{orcidlink}
\usepackage[nameinlink,noabbrev]{cleveref}

\newcolumntype{L}[1]{>{\raggedright\arraybackslash}p{#1}}
\newcolumntype{C}[1]{>{\centering\arraybackslash}p{#1}}

\tikzset{
  qdbactor/.style={
    draw,
    rectangle,
    minimum width=2.1cm,
    minimum height=0.45cm,
    inner sep=2pt,
    font=\footnotesize\bfseries
  },
  qdbbox/.style={
    draw,
    rectangle,
    inner sep=3pt,
    align=left,
    font=\footnotesize
  },
  qdblife/.style={line width=0.4pt},
  qdbflow/.style={-{Latex[length=1.8mm]}, line width=0.5pt},
  qdbphase/.style={draw, dashed, line width=0.4pt},
  qdbann/.style={font=\footnotesize, fill=white, inner sep=1pt},
  qdbbrace/.style={decorate, decoration={brace, amplitude=3pt}, line width=0.4pt}
}

\newcommand{\verifier}{\mathcal{V}}
\newcommand{\prover}{\mathcal{P}}
\newcommand{\adversary}{\mathcal{A}}

\newcommand{\sample}{\overset{\$}{\leftarrow}}
\newcommand{\negl}[1]{\operatorname{negl}(#1)}
\newcommand{\DF}{\mathrm{DF}}
\newcommand{\MF}{\mathrm{MF}}

\newcommand{\pDF}{p_{\DF}^{\star}}
\newcommand{\pMF}{p_{\MF}^{\star}}
\newcommand{\pDFCV}{p_{\DF}^{\mathrm{CV}}}
\newcommand{\pMFCV}{p_{\MF}^{\mathrm{CV}}}
\newcommand{\view}{v}
\newcommand{\AvgOmega}{\frac{1}{|\Omega|}\sum_{\omega\in\Omega}}

\theoremstyle{plain}
\newtheorem{theorem}{Theorem}

\newtheorem{lemma}[theorem]{Lemma}

\theoremstyle{definition}
\newtheorem{definition}{Definition}

\title{Security evaluation of quantum distance-bounding protocols via semidefinite programming}
\author{
  Kevin Bogner\,\orcidlink{0009-0003-5167-7177}$^{1}$,
  Aysajan Abidin\,\orcidlink{0000-0002-5128-3608}$^{1}$,
  Dave Singel\'ee\,\orcidlink{0000-0001-9084-698X}$^{2}$,
  Bart Preneel\,\orcidlink{0000-0003-2005-9651}$^{1}$\\[0.5em]
  \small $^{1}$COSIC, KU Leuven, Leuven, Belgium\\
  \small $^{2}$DistriNet, KU Leuven, Leuven, Belgium\\
  \small \texttt{kevin.bogner@kuleuven.be, aysajan@kuleuven.be,}\\
  \small \texttt{dave.singelee@kuleuven.be, bart.preneel@kuleuven.be}
}
\date{}

\begin{document}

\maketitle

\begin{abstract}
Quantum distance-bounding (QDB) protocols let a verifier check that a
prover is both genuine and physically nearby. During a timed fast phase
of quantum communication, the verifier measures round-trip times to
obtain an upper bound on the prover's distance.
Comparing the security of QDB protocols is difficult because the optimal
attack depends strongly on the protocol's structure. For a uniform
comparison, we isolate the fast phase and study one-round distance-fraud
(DF) and mafia-fraud (MF) games. For discrete-variable QDB, we show that
these games reduce to convex optimization problems and can therefore be
solved exactly with semidefinite programming; each MF value comes with an
explicit attack achieving it and a matching certificate that no attack
does better. This contrasts with quantum position verification, where an
attack is split between two separated parties, so its optimization is
nonconvex and analyses rely on relaxations. In our MF game, the
cooperating pair collapses to a single sequential strategy, which keeps
the game convex and its exact value computable. Across the
discrete-variable protocols we examine, the best one-round DF attack
succeeds with the same probability (\(1/2\)) for every protocol, whereas
MF clearly separates the protocols. For continuous-variable QDB, we
report estimated attack success probabilities from a calibrated Gaussian
attack model. The benchmark covers protocols whose fast phase itself
authenticates the prover; designs that follow Brands and Chaum and
instead bind the fast phase with a final authenticated message, like
the earliest QDB proposal, fall outside it and are treated separately. Of the four
protocols studied, two had no previously known one-round attack values,
and we report the first ones; for the other two, we find MF attacks
with higher success probability than previously reported. Overall, one-round MF resistance
depends on whether an attacker can use information revealed early by the
prover to answer a fresh challenge from the verifier.
\end{abstract}

\section{Introduction}

Distance-bounding (DB) protocols let a verifier check that a prover
knows the right secret key and that it is physically
nearby~\cite{brandschaum1994distance,hancke2005rfid}. Think of a car (the
verifier) checking that the genuine key fob (the prover) is actually near
the car, and not in the owner's home while an attacker relays its
signals. To this end, a DB protocol combines untimed setup messages (the
slow phase) with timed challenge-response rounds (the fast phase). The
round-trip time of each response limits how far away the prover can be. In quantum distance bounding (QDB), the
fast phase uses quantum communication, so the prover must receive,
process, and answer fresh quantum data within a strict time budget. This
has led to a growing family of QDB protocols, including
prepare-and-measure, entanglement-based, and continuous-variable QDB
(CV-QDB)
constructions~\cite{abidin2017towards,abidin2019quantum,abidin2025entanglement,bogner2024entangled,bogner2025continuous}.

As in classical DB, the main attack classes are distance fraud (DF),
mafia fraud (MF), and terrorist fraud (TF). In distance fraud, a
dishonest prover that holds a valid key but sits outside the allowed
distance tries to appear closer than it is. In mafia fraud, an attacker
without the key sits between an honest verifier and an honest prover and
relays or manipulates their messages; the relay attack on keyless car
entry sketched above is the classic example. In terrorist fraud, a
dishonest prover helps a nearby accomplice make the verifier accept,
without handing over its long-term key.

A central difficulty, made explicit by the QDB security framework
of~\cite{bogner2026}, is that QDB protocols are not easy to compare on common
ground. Their fast phases reveal different information at different
times, so the optimal attack can change substantially from one design to
another. Consequently, much of the literature focuses on protocol-specific
attacks, which are useful within a single protocol but do not give a
uniform comparison across QDB variants.

We close this gap by using the QDB security framework
of~\cite{bogner2026}, which formalizes the DF and MF games and the
attacker's capabilities. From it, we derive one-round fast-phase games
and optimize them via semidefinite programming (SDP). Throughout the
paper, the \emph{value} of such a game is the highest success probability
that any allowed attacker strategy can achieve in that single round. For
discrete-variable QDB (DV-QDB) protocols, the SDPs compute these values
exactly; for Gaussian CV-QDB we use a calibrated Gaussian estimator.

This exact solvability is the main conceptual point of the paper, and it
is best seen against quantum position verification
(QPV)~\cite{kent2011tagging,buhrman2014position}, the closest
well-studied setting, in which verifiers use timed quantum messages to
check a prover's location. A QPV attack is mounted by two separated
parties that intercept the verifiers' messages and must coordinate under
timing constraints. Optimizing such a split attack is in general a
nonconvex problem, so QPV analyses rely on relaxations matched by
explicit attacks. The QDB MF adversary is also a cooperating pair, but in
the one-round game the pair collapses to a single sequential strategy:
first query the honest prover, then answer the verifier's challenge
(\Cref{sec:mf-sdp}). This collapse keeps the game convex, so the optimal
attack value is not merely bounded but computed exactly. QDB thus offers
a position-verification-like setting in which one-round attack values
are exactly computable.

\subsection{Contributions}
\label{sec:contributions}

For DV-QDB, the reduced DF and MF games become SDPs. In the DF game, the
attacker simply prepares a quantum state, which gives a state SDP. In the
MF game, the attacker acts in two steps: it first queries the honest
prover before the timed challenge arrives (a \emph{pre-ask}) and then
answers the verifier's fresh challenge; such two-step strategies are
described by two-slot combs; this gives a comb SDP. Because both games
are convex, the SDP optimum is the exact value of the one-round game,
not a relaxation of it. Each MF value is certified by a matching primal
and dual pair in exact arithmetic: an explicit attack together with a
matching upper-bound certificate of the same value
(\Cref{app:mf-certificate-checks}). For Gaussian CV-QDB, exact SDPs are
out of reach; we instead report values from a calibrated Gaussian
estimator, which are estimates rather than certified values. Both sets
of results are summarized in \Cref{tab:fast-phase-summary}.

\subsection{Assumptions}
\label{sec:assumptions}

The results should be read under the following assumptions and scope
restrictions:

\begin{itemize}
\item One-round analysis.
We extract a single timed fast round from the full QDB execution and omit
correlations across rounds. The reported values are one-round fast-phase
values, not full-protocol fraud probabilities. \Cref{sec:conclusion}
discusses how these values relate to security over all \(n\) rounds.
Later protocol checks and implementation-level attacks are outside these
games.

\item Ideal randomness.
We model the secret values derived from the pseudorandom function (PRF) in the
slow phase as independent uniform bits. This idealization assumes that nonces are fresh and that the PRF is
secure against quantum polynomial-time (QPT) adversaries, i.e., efficient
quantum attackers.

\item Timing model.
We use the timing constraints of the QDB security framework~\cite{bogner2026}. In DF, the fast
response cannot depend on information from the verifier's actual challenge. In
MF, the response may use information obtained from the honest prover
before the timed verifier challenge.

\item Benchmark scope.
The benchmark covers protocols whose fast phase itself authenticates
the prover. QDB 2017~\cite{abidin2017towards} instead follows the
design of Brands and Chaum~\cite{brandschaum1994distance}: its fast phase
only checks proximity, and a final message authentication code (MAC)
authenticates the session. The one-round game strips that MAC, so we
treat QDB 2017 separately in \Cref{app:qdb2017}.

\item Protocol-specific reductions.
For the Mutual QDB protocol~\cite{abidin2025entanglement}, we report only
the first of its two timed exchanges (the first timed return leg). For
the E91 QDB protocol~\cite{bogner2024entangled}, we report only the
rounds in which verifier and prover measure in the same basis and the
outcome is checked (the value-check branches).

\item CV-QDB estimator.
For Gaussian CV-QDB~\cite{bogner2025continuous}, we use the restricted
Gaussian estimator of \Cref{sec:cvqdb} with zero added response noise,
matching the DV-QDB analysis and making the attacker as powerful as
possible. The reported MF value is an estimate from a finite parameter
grid, not a certified optimum; the DF value is the exact optimum within
this restricted Gaussian family.

\item Terrorist fraud.
TF security is about whether a malicious prover's help can be reused
across protocol runs, so it falls outside our one-round analysis (see
\Cref{sec:tf-sdp}).
\end{itemize}

\begin{table}[!t]
\centering
\footnotesize
\setlength{\tabcolsep}{4pt}
\caption{Values of the reduced one-round fast-phase games under the
assumptions of \Cref{sec:assumptions}; for the DV-QDB protocols, each
entry is the best success probability of the corresponding one-round
attack. The prior mafia-fraud attack column lists
previously reported per-round fast-phase attack values when a directly
comparable value is available. \(\text{--}\) marks entries with no comparable
prior per-round value; \(^{*}\) marks values exceeding previously reported
per-round fast-phase attacks; \(^{\dagger}\) marks values not previously
reported to our knowledge. The displayed DV-QDB decimals are rounded from
exact one-round values computed via SDP; the MF values are certified by
matching primal and dual witnesses in exact arithmetic
(\Cref{app:mf-certificate-checks}). The CV-QDB entries
assume zero added response noise;
the DF entry is the exact optimum within the restricted Gaussian model,
and the MF entry is a finite-grid Gaussian estimate. QDB
2017~\cite{abidin2017towards}, whose fast phase is bound by a final
MAC, falls outside this benchmark and is treated in
\Cref{app:qdb2017}.}
\label{tab:fast-phase-summary}
\vspace{0.35em}
\begin{tabular}{@{}lccc@{}}
\toprule
Protocol
& Distance-fraud value
& Prior mafia-fraud attack
& Mafia-fraud value \\
\midrule
\multicolumn{4}{@{}l}{\emph{Discrete-variable QDB protocols}}\\
QDB 2019~\cite{abidin2019quantum}
& \(0.5000\)
& \(0.875\)~\cite{verschoor2022quantum}
& \(0.9045^{*}\) \\
Mutual QDB~\cite{abidin2025entanglement}
& \(0.5000\)
& \(0.625\)~\cite{abidin2025entanglement}
& \(0.7500^{*}\) \\
E91 QDB~\cite{bogner2024entangled}
& \(0.5000^{\dagger}\)
& \(\text{--}\)
& \(0.9332^{\dagger}\) \\
\midrule
\multicolumn{4}{@{}l}{\emph{Continuous-variable QDB protocol}}\\
CV-QDB~\cite{bogner2025continuous}
& \(0.3592^{\dagger}\)
& \(\text{--}\)
& \(0.9138^{\dagger}\) \\
\bottomrule
\end{tabular}
\end{table}

\section{Related work}\label{sec:related-work}

DB was introduced by Brands and Chaum~\cite{brandschaum1994distance} and
later adapted into practical designs for radio-frequency identification
applications~\cite{hancke2005rfid}. Later work introduced the standard attack classes DF, MF, and TF, together
with a large body of protocol designs and analyses; a convenient overview is
given in~\cite{avoine2018survey}. Modeling TF is harder than DF or MF, and the literature has yet to agree
on a single definition of TF
resistance~\cite{fischlin2013terrorism,vaudenay2013modeling}.

Early QDB work proposed a qubit-based fast phase with a final classical
MAC; we call this protocol
QDB 2017~\cite{abidin2017towards}. QDB 2019~\cite{abidin2019quantum}
refined this by authenticating the prover through its quantum
measurement-and-preparation response, removing the MAC. This design
difference sets the scope of our benchmark (\Cref{sec:applications}).
More recent work
extended QDB to entanglement-based and continuous-variable settings: a
mutual entanglement-based protocol we call Mutual
QDB~\cite{abidin2025entanglement}, an E91-inspired variant we call E91
QDB~\cite{bogner2024entangled}, and a continuous-variable protocol we
call CV-QDB~\cite{bogner2025continuous}.

The work methodologically closest to ours is a cryptanalysis of early QDB
protocols~\cite{verschoor2022quantum}, which sharpened several attack
analyses, studied TF countermeasures and their limitations, and used SDP
to optimize measurements in protocol-specific attacks. We extend this idea
to a different scope: rather than applying SDP inside individual attacks,
we use it to capture the entire one-round strategy space. The analysis
rests on two foundations. First, the QDB security framework
of~\cite{bogner2026} formalizes the attacker's capabilities. Second,
from~\cite{mauw2018distance} we take the causal viewpoint that a party's
response can only depend on information that has had time to physically
reach it before it answers. Instead of tracking explicit times and
locations, we therefore only track which information is available to
each party at the moment it must respond. From these, we derive
one-round DF and MF games that we solve exactly for the DV-QDB protocols
and estimate for CV-QDB.

The closest related setting is QPV, in which verifiers use timed
quantum messages to check a prover's
location~\cite{kent2011tagging,buhrman2014position}. Unconditional QPV
is impossible against attackers that share unlimited
entanglement~\cite{buhrman2014position}, so QPV security is studied
against bounded attackers. One-round QPV games have been studied in
depth. The monogamy-of-entanglement game
of~\cite{tomamichel2013monogamy} gives closed-form attack values for
BB84-style QPV when the attackers share no entanglement beforehand,
\cite{bluhm2022position} analyzes single-qubit QPV under
multi-qubit attacks, and \cite{escolafarras2023loss} bounds
attack values for lossy QPV with SDP relaxations that are matched by
explicit attacks. Continuous-variable QPV has also
been analyzed~\cite{allerstorfer2023cvqpv,escolafarras2024cvqpv}. QDB
differs from QPV in two ways that change the attack model. First, the
QDB prover authenticates itself with a shared secret key, and the hidden
protocol variables in our games come from that key material. Second, the
QDB MF adversaries may query the honest prover before the timed
challenge (the pre-ask), and in the one-round game the cooperating pair
acts as a single sequential strategy (\Cref{sec:mf-sdp}), whereas a QPV
attack is mounted by two separated parties that intercept the verifiers'
messages and must coordinate under timing constraints using pre-shared
entanglement. Optimizing a QPV attack is in general a nonconvex problem
over separated strategies, so QPV analyses rely on relaxations and
explicit attacks, while the one-round DV-QDB games studied here
are convex and can be solved exactly as single SDPs.

\section{Preliminaries}\label{sec:prelims}

Our analysis rests on two standard ingredients, recalled in this section.
The first is the QDB security framework~\cite{bogner2026}, which sets the
rules of the attacks we study. The second is semidefinite programming, the
convex-optimization tool we use to solve the resulting games.

\subsection{QDB security framework}\label{sec:security-model}
We first recall the completeness definition from the QDB security
framework~\cite{bogner2026}. Completeness means that an honest prover
within the distance bound $B$ of the verifier is accepted with
overwhelming probability.

\begin{definition}[Completeness]\label{def:completeness}
    If an honest prover~\(\prover\) is located within the distance bound at distance \(d\le B\) from the verifier~\(\verifier\), where \(r_{\verifier}\) denotes the verifier's internal randomness, then
    \[
        \Pr\Bigl[
            (\verifier(x,r_{\verifier})\;\leftrightarrow\;\prover(x))
            \text{ accepts}
            \Bigr]
        \;\ge\; 1 - \negl{\lambda}.
    \]
\end{definition}

The framework also defines the DF and MF experiments, which the reduced
games below build on.

\begin{definition}[Distance-fraud experiment]
    \label{def:df-experiment}
    \noindent
    \begin{enumerate}
        \item Setup.
        The verifier~\(\verifier\) and a dishonest prover~\(\prover^{\star}\) share a long-term secret key~\(x\).
        The dishonest prover~\(\prover^{\star}\) is located outside the distance bound at a distance \(d>B\) from~\(\verifier\).
        
        \item Challenge session.
        \(\verifier\) and \(\prover^{\star}\) engage in a complete execution of the QDB protocol, which includes \(n\) fast rounds subject to the distance bound check enforced by~\(\verifier\).
        
        \item Output.
        \(\verifier\) outputs \(\mathrm{Out}_{\verifier}\in\{0,1\}\).
    \end{enumerate}
\end{definition}

\begin{definition}[Mafia-fraud experiment]
  \label{def:mf-experiment}
    \noindent
    \begin{enumerate}
        \item Setup.
              Verifier~\(\verifier\) and honest prover~\(\prover\) share a long-term
              secret key~\(x\).
              Two QPT adversaries,
              \(\adversary_1\) (co-located with~\(\verifier\)) and
              \(\adversary_2\) (co-located with~\(\prover\)),
              share a timing-constrained authenticated classical and quantum channel.
        
        \item Learning phase (Pre-ask).
              The adversaries may initiate and control any polynomial number of auxiliary executions of the QDB protocol between $\verifier$ and $\prover$:
              in each such execution, every classical and quantum message between the honest parties passes through $(\adversary_1,\adversary_2)$, who may relay, delay, drop, modify, or inject messages arbitrarily (subject only to the timing constraints).
              At the end of this phase, $(\adversary_1,\adversary_2)$ retain the entire classical transcript and their joint quantum state.
        
        \item Challenge session.
              The adversaries $(\adversary_1,\adversary_2)$ interact with $\verifier$ in a full execution of the QDB protocol.
              Simultaneously, they may interact with the honest $\prover$ (who uses the correct long-term key $x$ and is at distance $d > B$).
              $\verifier$ enforces the distance bound $B$.
        
        \item Output.
        $\verifier$ outputs \(\mathrm{Out}_{\verifier}\in\{0,1\}\).
    \end{enumerate}
\end{definition}

Each experiment is a game whose output is \(1\) when the verifier accepts
and \(0\) otherwise. The adversary's DF or MF advantage is its
probability of making the verifier accept. A QDB protocol is secure
against the corresponding fraud if this advantage is negligible in the
security parameter~\(\lambda\) for every QPT adversary that respects the
location and timing restrictions. Here \(\lambda\) sets the length of the
protocol's secret values (the nonces are \(\lambda\)-bit strings) and the
strength of its PRF and MAC. The full protocol runs \(n\) fast
rounds, and security must hold even against an adversary that attacks all
of them and runs any polynomial number of extra sessions.

\subsection{Semidefinite programs}\label{sec:sdp-prelim}

SDP methods are standard tools for optimizing quantum
strategies~\cite{gutoski2007toward,chiribella2009theoretical} and have
been used before to analyze attacks on QDB
protocols~\cite{verschoor2022quantum}.
For background, see \cite{boyd2004convex} on convex optimization and
\cite{watrous2018theory,skrzypczyk2023semidefinite} on SDPs in quantum
information, where a typical use is to compute the best success
probability that any quantum strategy can achieve in a given task. We
use SDPs in exactly this way: the strategies we optimize over are
attacks, so the computed value bounds what any attacker can do.

The quantities our programs optimize are probabilities, so we need
matrix tools that can never produce a negative number. A matrix
\(X\) is \emph{Hermitian} when it equals its own conjugate transpose;
its eigenvalues are then real. If in addition no eigenvalue is
negative, \(X\) is \emph{positive semidefinite} (PSD), written
\(X\succeq 0\); equivalently, \(\langle\psi|X|\psi\rangle\ge 0\) for
every vector \(|\psi\rangle\). Hermitian matrices are thus the matrix
analogue of real numbers, and PSD matrices the analogue of
nonnegative numbers. Our optimization variables will always be PSD
matrices, precisely to enforce this nonnegativity. Comparisons are defined as for numbers:
for Hermitian \(A\) and \(B\), we write \(A\succeq B\) when
\(A-B\succeq 0\). The \emph{trace} \(\operatorname{Tr}(M)\) is the
sum of the diagonal entries of \(M\); the trace of a product,
\(\operatorname{Tr}(AB)\), pairs the entries of \(A\) and \(B\) and
is the matrix analogue of the dot product of two vectors. One fact
is used repeatedly below, the analogue of the rule that a product of
two nonnegative numbers is nonnegative: if \(A\succeq 0\) and
\(B\succeq 0\), then \(\operatorname{Tr}(AB)\ge 0\). In quantum theory, this fact
guarantees that measuring a state never produces a negative
probability.

An SDP maximizes a linear function of a single PSD matrix variable,
subject to linear constraints:
\begin{equation}\label{eq:sdp-standard-form}
\max_{X\succeq 0}\ \operatorname{Tr}(CX)
\qquad\text{subject to}\qquad
\operatorname{Tr}(A_iX)=b_i \quad\text{for } i=1,\dots,m,
\end{equation}
where the Hermitian matrices \(C,A_1,\dots,A_m\) and the numbers
\(b_1,\dots,b_m\) are fixed data and the PSD matrix \(X\) is the
variable. In our programs the variable \(X\) is the attack strategy,
the one thing the adversary is free to choose, and the fixed data
describe the protocol under attack. The function being maximized,
called the \emph{objective}, is then the probability that the attack
succeeds. The constraints, together with \(X\succeq 0\), are the
conditions that make \(X\) a physical strategy, one that quantum
mechanics allows the adversary to implement: the linear constraints
normalize \(X\) so that its outcome probabilities sum to one, and
\(X\succeq 0\) keeps them nonnegative. Both the objective
and the constraints are linear in \(X\), because \(X\) enters them
only through traces against fixed matrices.

Our SDPs will not always look exactly like
\Cref{eq:sdp-standard-form}: some minimize instead of maximize, some
have inequality constraints instead of equalities, and some use
several matrix variables at once. Standard rewritings turn each of
these variations into the form above, so we state every program in
whatever form is most natural. Every SDP is \emph{convex}. A general
optimization problem can have local optima: solutions better than all
nearby alternatives but worse than the true optimum. A numerical
solver improves its answer step by step, so it can get stuck at such
a point and report it as the best. In a convex problem every local
optimum is the global one, so solvers converge reliably to the true
global optimum.

Quantum problems map to SDPs directly, because the basic objects of
quantum theory are PSD matrices with linear side conditions. A state
is a density matrix \(\sigma\) (\(\sigma\succeq 0\),
\(\operatorname{Tr}(\sigma)=1\)). A measurement with outcomes \(i\) is
a set of matrices \(\{E_i\}\), called a POVM, with each
\(E_i\succeq 0\) and \(\sum_i E_i=I\); outcome \(i\) occurs on state
\(\sigma\) with probability \(\operatorname{Tr}(E_i\sigma)\). These
numbers behave as probabilities must: each is nonnegative by the
trace fact above, and because \(\sum_i E_i=I\), they sum to
\(\operatorname{Tr}(\sigma)=1\). A
channel, and more generally a multi-step strategy, is again a PSD
matrix with linear conditions, in the Choi form introduced in
\Cref{sec:mf-sdp}. Probabilities like \(\operatorname{Tr}(E_i\sigma)\)
are traces of products and therefore linear in each object
separately: fix the measurement and the probability is linear in the
state; fix the state and it is linear in the measurement. Linear has
a concrete meaning here: prepare an equal mixture of two states, and
each outcome probability is the average of the two separate
probabilities. This is exactly the linearity an SDP needs.

Modeling a problem as an SDP therefore takes three steps: identify
what the optimizing party freely chooses, while everything the
protocol fixes becomes constant data; write each free choice as a PSD
variable together with the linear conditions that make it physical;
and check that the
success probability is linear in these variables. Only the last step
can fail. It holds whenever a single party optimizes, but when two
separated parties optimize jointly, their choices multiply inside the
objective and the problem is in general nonconvex. That is the QPV
obstacle discussed in \Cref{sec:related-work}: there too the
optimizing party is the adversary, but a QPV adversary is a pair of
attackers at two separated locations, exactly the situation where
choices multiply. Our MF adversary is also a cooperating pair, yet
our games avoid the obstacle, because in the one-round game the pair
collapses to a single sequential strategy.

We now carry out these three steps on a simple example:
minimum-error discrimination of two
states~\cite{helstrom1969quantum}. The task is a toy version of
mafia fraud: the adversary holds a challenge state that encodes a
bit it does not know and must guess that bit. A referee
flips a fair coin \(b\in\{0,1\}\) and sends the qubit state
\(\rho_b\), where \(\rho_0=|0\rangle\langle 0|\) and
\(\rho_1=|+\rangle\langle +|\) with
\(|+\rangle=(|0\rangle+|1\rangle)/\sqrt{2}\); the receiver measures
the qubit and guesses \(b\). The only free choice is the measurement,
a two-outcome POVM \((E_0,E_1)\) with outcome \(i\) meaning guess
\(b=i\), and the success probability is linear in \((E_0,E_1)\). The
optimal guessing probability is therefore the SDP value
\begin{equation}\label{eq:helstrom-primal}
p^{\star}
=
\max_{E_0,E_1}\ \tfrac12\operatorname{Tr}(E_0\rho_0)
+\tfrac12\operatorname{Tr}(E_1\rho_1)
\quad\text{subject to}\quad
E_0\succeq 0,\ E_1\succeq 0,\ E_0+E_1=I.
\end{equation}
Every \emph{feasible} point, a pair \((E_0,E_1)\) satisfying the
constraints, is a measurement the receiver can actually perform, so
feasible points are strategies, and each one lower-bounds
\(p^{\star}\). For these two states the optimum is
\(p^{\star}=\cos^{2}(\pi/8)\approx 0.854\), achieved by measuring in
the basis that makes an angle of \(\pi/8\) with each of the two
states~\cite{helstrom1969quantum}.

An upper bound on \(p^{\star}\) must hold for every measurement at
once; this is what the dual delivers. Take any Hermitian \(Y\) with
\(Y\succeq\tfrac12\rho_0\) and \(Y\succeq\tfrac12\rho_1\). For every
feasible \((E_0,E_1)\),
\[
\tfrac12\operatorname{Tr}(E_0\rho_0)+\tfrac12\operatorname{Tr}(E_1\rho_1)
\;\le\;
\operatorname{Tr}(E_0Y)+\operatorname{Tr}(E_1Y)
\;=\;
\operatorname{Tr}\bigl((E_0+E_1)Y\bigr)
\;=\;
\operatorname{Tr}(Y),
\]
where the inequality applies the trace fact above to \(E_i\) and
\(Y-\tfrac12\rho_i\). Each such \(Y\) is therefore a certificate: no
measurement whatsoever succeeds with probability above
\(\operatorname{Tr}(Y)\). Minimizing \(\operatorname{Tr}(Y)\) over all
such \(Y\) is itself an SDP, called the \emph{dual}; the original
maximization \Cref{eq:helstrom-primal} is the \emph{primal}. The bound
\(p^{\star}\le\operatorname{Tr}(Y)\) is called \emph{weak duality}.
The primal searches over strategies and pushes the value up from
below; the dual searches over certificates and presses it down from
above. In the example the two meet: the best certificate also has
value \(\cos^{2}(\pi/8)\). This is not special to the example: by
\emph{strong duality}, the two values coincide for every SDP that
satisfies a mild extra condition (Slater's condition), as all SDPs in
this paper do~\cite{boyd2004convex}. The dual is
constructed mechanically from the primal (here \(Y\) corresponds to
the constraint \(E_0+E_1=I\)), and \Cref{app:mf-certificate-checks} states
the dual of our MF program, built in the same way.

In practice, standard software first supplies candidate primal and dual
solutions. One states the primal in a few
lines of code, essentially as written in \Cref{eq:helstrom-primal},
and a solver computes the optimum, to a given accuracy, in time
polynomial in the matrix dimension and the number of constraints. Two
features of the solver output matter here. First, the solver
returns a primal and a dual solution together, so every computed value
arrives with a candidate optimal strategy and a candidate certificate.
Second, the output is floating point and satisfies the constraints only
up to a small tolerance, so on its own it is numerical evidence, not a
proof. For the MF values, we reconstruct explicit algebraic primal and
dual witnesses and verify them in exact arithmetic. The verifier rebuilds
the SDP tester and checks that both witnesses are feasible and have the
same objective value, without calling the solver or reading its
floating-point output (\Cref{app:mf-certificate-checks}). The SDPs in
this paper are small by solver standards: the largest variable, the
MF strategy matrix of \Cref{sec:mf-sdp}, is a \(16\times 16\) matrix.

Our DV-QDB fast-round games fit this three-step pattern. The classical data in a
round (the challenge bits, key bits, basis choices, and measurement
outcomes) take only finitely many values, and the quantum messages are
qubits rather than infinite-dimensional systems, so the whole game can
be written with finite-size matrices. Only the adversary's side is
free. In DF it chooses the response states it returns; in MF it chooses
one two-step strategy, captured by a single PSD matrix whose
normalization conditions are linear. The verifier's acceptance
probability is linear in these variables. The DF and MF problems in the
next section are therefore SDPs, and each solved value comes with a
strategy that achieves it and a certificate that nothing does better.

\section{Reduction of one-round fast-phase games to SDP}
\label{sec:sdp-reduction}

We use the security framework of \Cref{sec:security-model} to determine what the
adversary can do in the one-round DF and MF games, in particular what
information it may use and when. Following~\cite{mauw2018distance}, we
track this information rather than explicit times and locations. For
DV-QDB, these games reduce to exact SDPs. The DF game
becomes a state SDP, since the distant prover only prepares a quantum
response state. The SDP finds the state that maximizes the verifier's
acceptance probability.
The MF adversary, by contrast, is a cooperating pair that we treat as a
single two-step strategy. It pre-asks the honest prover, keeps the
returned quantum system in memory, then processes it together with the
verifier's challenge to produce its response. This strategy is a two-slot
comb, so the MF game becomes a comb SDP that maximizes the verifier's
acceptance probability.

For the DV reductions in this section, we describe a one-round game instance
by a finite set of branches \(\Omega\), each equally likely, so a branch has
weight \(1/|\Omega|\). A branch \(\omega\in\Omega\) is one possible setting of
the round's classical variables. Crucially, some of these variables must stay
hidden from the adversary. Otherwise the adversary is able to produce the honest
prover's response, which the verifier always accepts by completeness
(\Cref{def:completeness}). On branch \(\omega\), the verifier sends the
challenge state \(\rho_\omega\) and receives the returned state \(\sigma\)
(honest or adversarial). The accept operator \(\Pi_\omega\), the verifier's
check, is determined by the branch. The verifier accepts with probability
\(\operatorname{Tr}(\Pi_\omega\sigma)\), which measures how close \(\sigma\) is
to the response it expects on this branch. The bound
\(0\preceq \Pi_\omega\preceq I\), with \(I\) the identity matrix, is what makes this a genuine probability
between \(0\) and \(1\). Without it, \(\Pi_\omega\) would be an arbitrary matrix
rather than a real measurement.

The CV-QDB protocol is handled separately
in \Cref{sec:cvqdb}. We do not formulate a TF SDP. \Cref{sec:tf-sdp}
explains why.

\subsection{Distance-fraud reduction to SDP}\label{sec:df-sdp}

We reduce the DF experiment of \Cref{def:df-experiment} to a single round. Following~\cite{mauw2018distance}, let
$\view(\omega)$ denote the DF response-time view, the information available to
the distant prover $\prover^{\star}$ at the moment it must respond. Typically
the view contains the secret key and everything from the slow phase, but not
the fast-phase challenge state. In distance
fraud, $\prover^{\star}$ lies outside the distance bound, so it must answer
before the challenge state $\rho_\omega$ could even reach it. Its response
therefore cannot depend on $\rho_\omega$, and the returned state depends only
on $\view(\omega)$. This restriction is formalized in~\cite{bogner2026, boureanu2015practical}. The optimal DF success probability is therefore
\begin{equation} \label{eq:df-sdp}
\pDF = \max_{\{\sigma_v\}} \AvgOmega \operatorname{Tr}\bigl(\Pi_\omega \sigma_{\view(\omega)}\bigr)
\end{equation}
subject to
\[
\sigma_v \succeq 0, \qquad \operatorname{Tr}(\sigma_v)=1 \quad \text{for every view } v.
\]
Here $\{\sigma_v\}$ is the prover's strategy, with $\sigma_v$ the state it returns when it has view $v$. The two constraints make every $\sigma_v$ a valid quantum state, positive semidefinite and with trace one, so the prover could actually prepare it. On branch $\omega$, the prover has view $v = \view(\omega)$ and returns $\sigma_{\view(\omega)}$, which the verifier accepts with probability $\operatorname{Tr}(\Pi_\omega \sigma_{\view(\omega)})$. Averaging this over all branches and choosing the best response states gives the optimal DF success probability $\pDF$.

\begin{lemma}[Reduced distance-fraud value]
For any reduced one-round DF instance of a DV-QDB protocol in the model above, the optimal DF value is given by \Cref{eq:df-sdp}.
\end{lemma}

\begin{proof}
Fix a response-time view $v$. By the timing restriction, the returned system cannot depend on the challenge state $\rho_\omega$. It is therefore described by some density operator $\sigma_v$. On branch $\omega$, the verifier accepts the returned state with probability $\operatorname{Tr}(\Pi_\omega \sigma_{\view(\omega)})$; averaging these acceptance probabilities over the equally likely branches gives the attack's success probability, which is exactly the objective of \Cref{eq:df-sdp}. Conversely, any feasible family $\{\sigma_v\}$ is physically realizable, since on observing view $v$ the prover can prepare $\sigma_v$. Optimizing over feasible density operators therefore yields the optimal DF value.
\end{proof}

\subsection{Mafia-fraud reduction to SDP}\label{sec:mf-sdp}
We reduce the MF experiment of \Cref{def:mf-experiment} to a single round. In that experiment, the adversary is the cooperating pair \((\adversary_1,\adversary_2)\), located near
the verifier and prover, respectively. In the one-round SDP game we replace the pair
\((\adversary_1,\adversary_2)\) by a single adversary \(S\) that acts in two
steps and uses the same strategy on every branch. Formally, \(S\) is a
branch-independent two-slot comb.

The first step is the pre-ask, the learning phase of \Cref{def:mf-experiment}.
The adversary chooses a probe system \(Q\) and
sends it to the honest prover, which returns a response \(R\). This reply maps
the input \(Q\) to the output \(R\), so we describe it by a quantum channel
\[
\mathcal N_\omega:\mathcal L(Q)\to\mathcal L(R),
\]
the most general physical process taking one quantum system to another;
here \(\mathcal L(Q)\) denotes the matrices acting on register \(Q\). The
channel is branch-dependent because the honest prover uses protocol variables
that differ from branch to branch. The adversary, by contrast,
does not know \(\omega\), so the same comb \(S\) is used on every branch.
The second step mirrors the challenge session of \Cref{def:mf-experiment}. The
verifier sends the challenge state \(\rho_\omega\) on a register \(C\) (a system
carrying one message), and the adversary returns its output on a register
\(O\). This output depends on the challenge \(\rho_\omega\) together with the
pre-ask systems \(Q\) and \(R\). The comb \(S\) is defined on the fixed
registers \(Q,R,C,O\), and the branch \(\omega\) selects which challenge state
arrives on \(C\). \Cref{fig:mf-two-slot-comb} shows the comb strategy \(S\) and its interactions with the prover and verifier.

\begin{figure}[t]
\centering
\begin{tikzpicture}[font=\footnotesize]

  \node[qdbactor] (V) at (-3.45,0)
    {Verifier};
  \node[font=\footnotesize\bfseries] (A) at (0,0)
    {Adversary};
  \node[qdbactor] (P) at (3.45,0)
    {Prover};

  \draw[qdblife] (V.south) -- (-3.45,-2.70);
  \draw[qdblife] (P.south) -- (3.45,-2.70);
  \draw[line width=1.2pt] (-3.70,-2.70) -- (-3.20,-2.70);
  \draw[line width=1.2pt] (3.20,-2.70) -- (3.70,-2.70);
  \draw[qdbphase, rounded corners=2pt]
    (-1.10,-0.42) rectangle (1.10,-2.70);
  \node[align=center] at (0,-1.56)
    {comb \(S\)};

  \draw[qdbflow]
    (1.10,-0.95)
    -- node[above] {\(Q\)}
    (3.45,-0.95);

  \draw[qdbflow]
    (3.45,-1.22)
    -- node[below] {\(R\)}
    (1.10,-1.22);
  \node[anchor=west] at (3.60,-1.09)
    {\(\mathcal N_\omega\)};

  \draw[qdbflow]
    (-3.45,-2.08)
    -- node[above] {\(C:\rho_\omega\)}
    (-1.10,-2.08);

  \draw[qdbflow]
    (-1.10,-2.35)
    -- node[below] {\(O\)}
    (-3.45,-2.35);
  \node[anchor=east] at (-3.60,-2.22)
    {\(\Pi_\omega\)};

\end{tikzpicture}
\caption{Reduced mafia-fraud SDP game. The adversary's comb \(S\) first
interacts with the prover via the pre-ask channel \(\mathcal N_\omega\), then
answers the verifier's challenge \(\rho_\omega\). The branch-dependent triple
\((\mathcal N_\omega,\rho_\omega,\Pi_\omega)\), with \(\Pi_\omega\) the
verifier's accept operator, combines into a single tester operator
\(W_\omega\) that captures everything outside the adversary's control on
branch \(\omega\). The acceptance probability is then
\(\operatorname{Tr}(W_\omega S)\), a linear function of the comb \(S\); this
linearity is what makes the optimization an SDP.}
\label{fig:mf-two-slot-comb}
\end{figure}
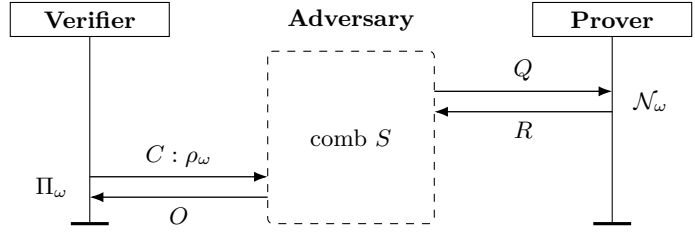

To turn the game into an SDP, we now put the comb \(S\) and the channels in
matrix form, using the unnormalized Choi convention
\[
J(\mathcal E)
=
\sum_{i,j}|i\rangle\langle j|\otimes
\mathcal E(|i\rangle\langle j|)
\]
with respect to the input computational basis. This represents every quantum
channel as a single positive semidefinite matrix on the tensor product of its
input and output spaces, and it turns the comb \(S\) into an operator on
\(Q\otimes R\otimes C\otimes O\). The MF acceptance probability then becomes a
linear function of \(S\), which is exactly what an SDP optimizes. The
comb-normalization constraints on \(S\) are imposed in the SDP below.

For each branch \(\omega\), the three branch-dependent objects
\((\mathcal N_\omega, \rho_\omega, \Pi_\omega)\) combine into a single tester
operator \(W_\omega\). The tester is the protocol-side counterpart of the
adversary's comb. It packages everything outside the adversary's control on
branch \(\omega\), and pairs with \(S\) via the trace so that
\(\operatorname{Tr}(W_\omega S)\) is exactly the acceptance probability. With
the Choi convention above, the tester reads
\[
W_\omega
=
J(\mathcal N_\omega)^{\mathsf T}_{QR}
\otimes
\left(\rho_\omega\right)^{\mathsf T}_{C}
\otimes
\Pi_{\omega,O},
\]
where transposes are taken in the computational basis of each register.

Thus, the optimal MF success probability is
\begin{equation}
\label{eq:mf-sdp}
\pMF
=
\max_{S,\eta}
\AvgOmega
\operatorname{Tr}(W_\omega S),
\end{equation}
subject to
\[
S \succeq 0,
\qquad
\operatorname{Tr}_{O}(S)=\eta \otimes I_{RC},
\qquad
\eta \succeq 0,
\qquad
\operatorname{Tr}(\eta)=1.
\]
Here \(\eta\) is the probe quantum state that the adversary prepares on
register \(Q\) and sends to the honest prover in the pre-ask interaction,
and \(\operatorname{Tr}_{O}\) is the partial trace over \(O\): it
discards register \(O\) and returns the state of the remaining
registers.
The constraints above are the standard comb normalization conditions; they
ensure that \(S\) describes a physically realizable sequential strategy in
the order \(Q,R,C,O\).

To prove each MF value, we use a companion dual SDP whose feasible points
upper-bound \(\pMF\). For every reported instance, explicit primal and dual
witnesses have the same algebraic objective value, which pins \(\pMF\) to that
value. Their feasibility and objective equality are rechecked in exact
arithmetic, so the proof does not rely on the solver. Details are in
\Cref{app:mf-certificate-checks}.

\begin{lemma}[Reduced mafia-fraud value]
For any reduced one-round MF instance of a DV-QDB protocol in the model above,
the optimal MF value is given by \Cref{eq:mf-sdp}.
\end{lemma}

\begin{proof}
An MF adversary first prepares a probe, receives the honest-prover
reply, later receives the verifier challenge, and then outputs its response.
Such sequential strategies are exactly deterministic combs on the
registers \(Q,R,C,O\) satisfying the constraints in \Cref{eq:mf-sdp}. For each
branch \(\omega\), the honest-prover channel, challenge state, and accept
operator combine into the tester \(W_\omega\), and pairing it with the comb
\(S\) gives the branch acceptance probability \(\operatorname{Tr}(W_\omega S)\).
The reverse also holds. Every feasible \(S\) corresponds to a real adversary
strategy~\cite{gutoski2007toward,chiribella2009theoretical}. Averaging over branches gives \Cref{eq:mf-sdp}.
\end{proof}

\subsection{Why TF is not reduced to SDP}\label{sec:tf-sdp}
Our SDPs model only a single fast round, which is enough for the DF and MF
games studied here. We do not formulate a TF SDP, because TF security is about
what happens across multiple runs: whether a malicious prover's help can be
reused in a fresh execution. A single-round SDP cannot capture that.

\section{Applying SDP to quantum distance bounding}
\label{sec:applications}

We now instantiate the one-round DF and MF SDPs of
\Cref{sec:sdp-reduction} for each QDB protocol. We treat the prepare-and-measure protocol
QDB 2019~\cite{abidin2019quantum}, the
entanglement-based protocols (Mutual QDB~\cite{abidin2025entanglement},
E91 QDB~\cite{bogner2024entangled}), and the CV
protocol~\cite{bogner2025continuous}. We list only the SDP inputs; protocol-flow diagrams are in
\Cref{app:protocol-flows}. For each DV-QDB protocol,
\Cref{tab:dv-instantiations} lists the branch set, the DF response-time view,
the verifier challenge state, the acceptance operator, the honest-prover
pre-ask channel, and the resulting one-round values. These values
should be read under the assumptions stated in
\Cref{sec:assumptions}.

One early protocol is deliberately absent from this benchmark. QDB
2017~\cite{abidin2017towards} follows the classical design of Brands
and Chaum~\cite{brandschaum1994distance}: its timed fast phase only
checks that the prover is close, and a final MAC over the fast-phase
transcript checks that it is genuine. The one-round game strips that
MAC, and with it the very mechanism that provides MF resistance, so
the resulting value (\(\pMF=1\)) reflects the missing MAC rather than
the protocol. The same holds for any QDB protocol that binds its fast phase
with a later authenticated message. The benchmark below therefore
covers the protocols whose fast phase itself authenticates the prover;
\Cref{app:qdb2017} details the QDB 2017 reduction and its
fast-phase-only values.

Throughout this section, we use the following notation:
\begin{itemize}
\item \(a,b\) denote fast-phase preparation or measurement bases, \(c\) a
challenge bit, \(m\) the verifier's measurement outcome in the
entanglement-based and CV protocols, and \(r\) a masking value.
\item The accept operator \(\Pi_\omega\) and channel \(\mathcal N_\omega\)
carry the full branch \(\omega\) as a subscript. For each protocol we instead
index them by the relevant variables, writing for instance \(\mathcal N_a\) if
the channel depends only on \(a\), and \(\mathcal N_{a,b}\) if it depends on
both \(a\) and \(b\).
\item For the BB84-style prepare-and-measure bases, we write
\(|u\rangle_a:=H^a|u\rangle\) for \(u,a\in\{0,1\}\), where \(H\) is the
Hadamard transform. Thus \(a=0\) gives the computational basis and \(a=1\)
the Hadamard basis.
\item For E91 QDB, the verifier and prover settings \(a_v\) and \(a_p\) each
select one of three measurement bases (\(X,W,Z\) for the verifier and
\(W,Z,V\) for the prover, defined in the E91 QDB subsection below), rather
than a BB84 basis. The response basis \(b\) is still BB84-style.
\item For the entanglement-based protocols, \(|\chi_{a,m}\rangle\) is the
reduced-game challenge state. When the verifier measures its half of the
entangled pair in basis \(a\) and gets outcome \(m\), the prover's half
collapses to \(|\chi_{a,m}\rangle\). For the perfectly correlated
\(|\Phi^{+}\rangle\) pairs of Mutual QDB, this is \(|m\rangle_a\). For the
perfectly anti-correlated singlet pairs of E91 QDB, it is the \(a_v\)-basis
state orthogonal to \(|m\rangle_{a_v}\).
\end{itemize}

\begin{table}[t]
\centering
\footnotesize
\setlength{\tabcolsep}{3pt}
\caption{DV-QDB reduced one-round SDP instantiations. All bit variables are
uniformly distributed, and all sums over the prover's measurement outcome \(t\) are over \(\{0,1\}\). In the
\(\rho_\omega\to\Pi_\omega\) column, a ket \(|\psi\rangle\) stands for
\(|\psi\rangle\langle\psi|\) (a pure state on the left of the
arrow, an accept projector on the right). For E91 QDB, the row covers only the value-check pairs
\((a_v,a_p)\in\{(1,0),(2,1)\}\); the other settings are used for CHSH testing
or left unchecked.}
\label{tab:dv-instantiations}
\begin{tabular}{@{}lcccc@{}}
\toprule
Protocol
& \(\omega;\,\view(\omega)\)
& \(\rho_\omega \to \Pi_\omega\)
& \(\mathcal N_\omega(\rho)\)
& \((\pDF,\pMF)\) \\
\midrule
QDB 2019
& \((a,b,c);\, (a,b)\)
& \(|c\rangle_a \to |c\rangle_b\)
& \(\mathcal N_{a,b}(\rho)=
   \sum_t \langle t|_a\rho|t\rangle_a
   |t\rangle_b\langle t|\)
& \((0.5000,\,0.9045)\)
\\
Mutual QDB
& \((a,b,r,m);\, (a,b,r)\)
& \(|\chi_{a,m}\rangle \to |m\oplus r\rangle_b\)
& \(\mathcal N_{a,b,r}(\rho)=
   \sum_t \langle t|_a\rho|t\rangle_a
   |t\oplus r\rangle_b\langle t\oplus r|\)
& \((0.5000,\,0.7500)\)
\\
E91 QDB
& \((a_v,a_p,b,m);\, (a_p,b)\)
& \(|\chi_{a_v,m}\rangle \to |\tilde m\rangle_b\)
& \(\mathcal N_{a_p,b}(\rho)=
   \sum_t \langle t|_{a_p}\rho|t\rangle_{a_p}
   |t\rangle_b\langle t|\)
& \((0.5000,\,0.9332)\)
\\
\bottomrule
\end{tabular}
\end{table}

Two patterns stand out from the DV-QDB values in
\Cref{tab:dv-instantiations}. First, the DF value is \(1/2\) for every DV-QDB protocol considered
here, so DF does not distinguish them. Second, MF provides the more meaningful
comparison. A protocol has a high MF value when the adversary can extract
information from the pre-ask that helps it answer the verifier's later
challenge, and a lower value when it cannot. The DV MF values are SDP-certified by matching primal (lower-bound) and dual
(upper-bound) witnesses; details are in \Cref{app:mf-certificate-checks}.

\subsection{QDB 2019}

QDB 2019~\cite{abidin2019quantum}, the most studied QDB protocol, uses
independent bases for the verifier's challenge \(a\) and the prover's response
\(b\). The protocol flow is shown in \Cref{fig:qdb2019-flow}. The instantiation
in \Cref{tab:dv-instantiations} gives \(\pMF\approx0.9045\) (certified in
\Cref{app:mf-certificate-checks}).

The previously reported per-round attack~\cite{verschoor2022quantum} attains
\(\pMF=0.875\) by reducing the pre-ask phase to a classical estimate of the
relation between \(a\) and \(b\). Our SDP attack is stronger. Instead of
collapsing the pre-ask to such an estimate, the adversary keeps its quantum
state intact and processes it jointly with the later challenge.

\subsection{Mutual QDB}

Mutual QDB~\cite{abidin2025entanglement} is an entanglement-based protocol in
which the verifier and prover authenticate each other. The verifier distributes
EPR pairs, and the prover returns measurement outcomes masked by a committed
value \(r\). The protocol flow is shown in \Cref{fig:mutualqdb-flow}. We analyze
only the first timed return leg, since our model covers a single fast round. The
second leg and the commitment opening are out of scope. The instantiation in
\Cref{tab:dv-instantiations} gives \(\pMF=0.75\).

The previously reported per-round attack~\cite{abidin2025entanglement} attains
\(\pMF=0.625\). Our SDP attack is stronger, and its value follows from a simple
basis-matching strategy. The adversary pre-asks the prover with
\(|\hat m\rangle_{\hat a}\) and stores the returned state. If \(\hat a=a\), the
stored response is \(|\hat m\oplus r\rangle_b\).
The later verifier challenge is the entangled-half state carrying the fresh
outcome \(m\) in basis \(a\), so the adversary can measure it in basis \(\hat a\) to learn \(m\). It
then flips the stored response iff \(m\ne\hat m\), using a Pauli \(Y\) operation,
which flips the logical bit in both BB84 bases up to phase, and obtains
\(|m\oplus r\rangle_b\). If \(\hat a\ne a\), the same procedure is uncorrelated
with the fresh outcome and succeeds with probability \(0.5\). This gives
\(0.5\cdot1+0.5\cdot0.5=0.75\), matching the SDP value.

\subsection{E91 QDB}

E91 QDB~\cite{bogner2024entangled} is a Bell-inequality-based entanglement
protocol. The verifier prepares an entangled pair, measures its own half in
setting \(a_v\), and sends the other half to the prover. The prover measures it
in setting \(a_p\) and returns the outcome. Settings \(0,1,2\) select the bases
\(X, W, Z\) for the verifier and \(W, Z, V\) for the prover, where \(W\) and
\(V\) are the eigenbases of the observables
\((Z+X)/\sqrt{2}\) and \((Z-X)/\sqrt{2}\). Their Bloch-sphere directions
lie at \(+\pi/4\) and \(-\pi/4\) from \(Z\), so their basis vectors are
rotated by \(\pm\pi/8\) from the computational-basis vectors. The two matching-basis
pairs \((a_v, a_p) \in \{(1,0), (2,1)\}\) are then exactly the rounds in which
verifier and prover measure in the same basis (\(W\) and \(Z\), respectively),
and these are used for value checks. The other settings serve a CHSH
non-locality test or are unchecked. The protocol flow is shown in
\Cref{fig:e91qdb-flow}. We analyze only the value-check pairs.
The instantiation in \Cref{tab:dv-instantiations} gives
\(\pMF \approx 0.9332\). To our knowledge, both the DF value
(\(\pDF=0.5\)) and the MF value reported here are the first per-round values
for this protocol.

For the value-check pairs, we use the singlet anti-correlation convention of
the SDP instantiation. The target outcome is \(\tilde m_i = 1 - m_i\), with
\(m_i\) the verifier's measurement outcome. The earlier value-matching
convention of~\cite{bogner2024entangled} sets \(\tilde m_i = m_i\). The two
conventions are equivalent up to a relabeling of the prover's classical bit.

The high MF value comes from the pre-ask interaction. The prover measures
the adversary's probe in basis \(a_p\) and re-encodes the outcome in the
response basis \(b\), so the returned state carries strong information
about both bases. The adversary processes this state jointly with the
fresh challenge state in basis \(a_v\) to produce the accepted response.

\subsection{Continuous-variable QDB}\label{sec:cvqdb}

The CV-QDB protocol we consider is the Gaussian protocol
of~\cite{bogner2025continuous}, shown in \Cref{fig:cvqdb-flow}. Unlike the
DV-QDB protocols, its acceptance test is threshold-based. The verifier accepts a
round when the returned value lies within a tolerance \(\delta_{\mathrm{cv}}\)
of its reference value. In the DV protocols the binary
response lets a distant prover succeed at least half the time, so the DF value
never drops below \(0.5\). Here the returned value must instead land close to a
continuous target, so a distant prover often misses and the DF value can fall
below \(0.5\). A tighter \(\delta_{\mathrm{cv}}\) lowers it further. An exact
optimum is out of reach here. The CV state space is infinite-dimensional, so
the attack optimization does not reduce to a finite-dimensional SDP as it does
for the DV protocols. We therefore restrict to Gaussian attacks and report the
estimators \(\pDFCV\) and \(\pMFCV\). To our knowledge, these are the first
per-round DF and MF values for this protocol.

The DV-QDB values are computed in a noiseless setting. To make the
CV estimator comparable, we calibrate the parameters to an honest one-round
completeness of about \(0.99\). The squeezing \(r_{\rm sq}\) controls the
strength of the EPR correlation between the verifier's and prover's modes, so a
larger value ties the honest response more tightly to the reference and raises
completeness. The \(\hbar=2\) convention fixes the vacuum variance at \(1\), so
the quadratures and the acceptance threshold are expressed in units of the
vacuum (shot) noise. At this operating point we use \(r_{\rm sq}=1.2\) and
\(\delta_{\mathrm{cv}}=1.1\), with zero added noise in the measured response
quadrature. \Cref{fig:cv-calibration-sweep} shows how completeness and the DF
and MF values respond as these two parameters vary around the operating point.
Panel~(a) sweeps the threshold and panel~(b) the squeezing.

\subsubsection{CV distance fraud}
In the DF game, the distant prover must answer before the challenge mode
arrives, so its response cannot depend on the verifier's reference value. For
each branch \((a,b)\), let \(m_i\) be the verifier's reference homodyne value
and let \(\operatorname{Var}_{\mathrm{resp}}(b)\) be the smallest response
variance the prover can achieve in basis \(b\). Because the response and the
reference are independent, the error variance is the sum of the two,
\[
\mathrm{MSE}^{\mathrm{DF}}_{a,b}
=
\operatorname{Var}(m_i)+\operatorname{Var}_{\mathrm{resp}}(b).
\]
A round passes when the returned value lies within \(\delta_{\mathrm{cv}}\) of
the reference. With a Gaussian error, each branch passes with an
error-function probability, and averaging over the four equally likely branches
gives
\[
\pDFCV
=
\frac14\sum_{a,b}
\operatorname{erf}\left(
\frac{\delta_{\mathrm{cv}}}
{\sqrt{2\,\mathrm{MSE}^{\mathrm{DF}}_{a,b}}}
\right).
\]
At the calibrated operating point this gives \(\pDFCV \approx 0.3592\).

\subsubsection{CV mafia fraud}
The CV MF model has the same two-slot structure as the DV MF, a pre-ask
interaction followed by the verifier's challenge. The adversary sends a Gaussian
probe to the honest prover and measures the returned mode. It then combines this
pre-ask information with a measurement of the verifier's later challenge mode to
estimate the verifier's target value. Unlike the DV MF, the CV version has no
exact comb SDP. A Gaussian attack is fixed by a few continuous parameters, so we
scan a finite grid of these parameters and keep the strongest attack. We add no
response noise, an attacker-favorable choice that can only raise the estimated
attack success. We bound the probe energy by \(N_{\rm pre}\le2\), since an
unbounded probe could read out the hidden basis and pass trivially. At the calibrated operating point
this gives \(\pMFCV(N_{\rm pre}\le 2) \approx 0.9138\). This value is a
finite-grid Gaussian estimator, not a certified finite-energy attack
probability. The grid and estimator details are in \Cref{app:cv-grid-benchmark}.

\begin{figure}[tbp]
  \centering
  \includegraphics[width=\linewidth]{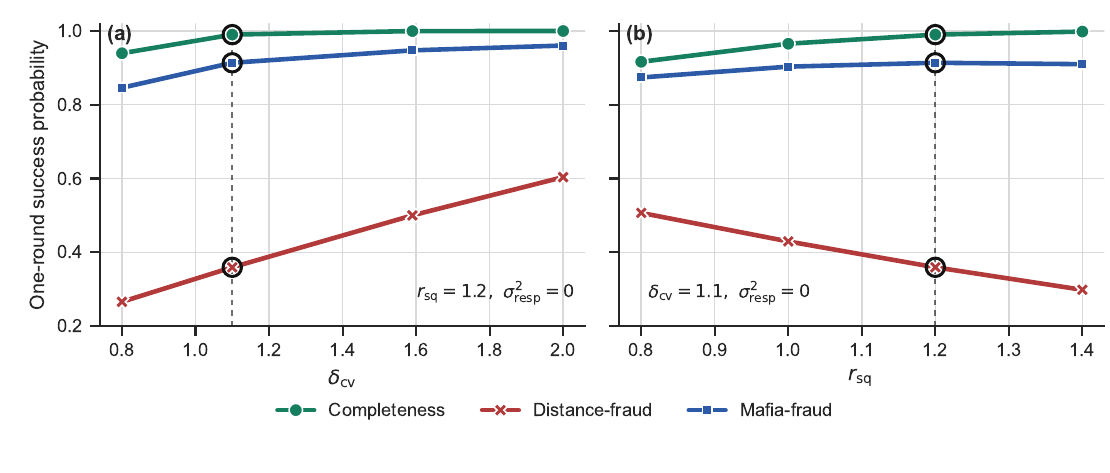}
  \caption{CV-QDB calibration sweeps with zero added response noise
  \(\sigma_{\mathrm{resp}}^2=0\). Panel~(a)
  varies the threshold at fixed \(r_{\rm sq}=1.2\), while panel~(b) varies the
  squeezing at fixed \(\delta_{\mathrm{cv}}=1.1\). Mafia-fraud values are finite-grid
  Gaussian estimators recomputed for each parameter point. Black outlines
  identify the calibrated operating point used in
  \Cref{tab:fast-phase-summary}.}
  \label{fig:cv-calibration-sweep}
\end{figure}

\section{Conclusion}\label{sec:conclusion}
We introduced a uniform one-round fast-phase framework for QDB and showed that
the DV-QDB DF and MF games can be solved exactly as semidefinite programs, a
state SDP for DF and a two-slot comb SDP for MF. Unlike in QPV, where attack
optimization is nonconvex and analyses rely on relaxations, these one-round
games are convex, so the reported DV values are exact rather than bounds.
This gives a single way to compare protocols across families, rather than a
separate attack tailored to each one. For QDB 2019 and the first timed return leg of Mutual QDB, the SDP
yields higher one-round MF values than previously reported. For E91 QDB and CV-QDB, we report the
first per-round DF and MF values.

For every DV-QDB protocol, \(\pDF = 0.5\), so DF does not distinguish between
them. The MF values instead range from \(0.75\) to \(0.9332\). What matters for MF
resistance is not whether a protocol is prepare-and-measure or
entanglement-based, but whether an early interaction with the prover can be
converted into a valid response to a fresh verifier challenge.
Protocols in the style of Brands and Chaum, such as QDB 2017, whose MF
resistance comes from a final MAC rather than from the fast phase
itself, fall outside this comparison (\Cref{app:qdb2017}).

The CV-QDB results point in the same direction, with one caveat. Because
acceptance is threshold-based, the values \(\pDFCV\) and \(\pMFCV\) depend on
the chosen squeezing, threshold, energy bound, grid, estimator class, and noise
model.

These results are reduced one-round fast-phase values, not full-protocol
fraud probabilities. A full-protocol adversary attacks all \(n\) rounds
with a single strategy. It can keep quantum memory between rounds, adapt
to earlier rounds, and interleave its pre-asks and responses, so the
\(n\)-round fraud probability cannot simply be assumed to be the
one-round value raised to the power \(n\). How repeated-game values
relate to one-round values is
the parallel-repetition question, which is subtle even for the
most-studied games. The repeated value can decay more slowly than the
\(n\)-th power~\cite{raz2011counterexample}, and for games with entangled
players the known general bounds are weaker
still~\cite{yuen2016parallel}. The QDB security
framework~\cite{bogner2026} provides the multi-round side of the
analysis; one route there is a concentration bound that tolerates
adaptive adversaries and turns a per-round success cap, if one holds
conditionally on the adversary's history, into a soundness bound for
threshold acceptance. The exact one-round values reported here are the
natural per-round constants for such a multi-round lifting, and for
QDB 2019 the stronger attack of \Cref{tab:fast-phase-summary} updates
the per-round value used there. Whether these one-round values cap the
adversary's success in every round of a full execution is a
protocol-by-protocol question beyond our one-round scope.

A natural next step is to feed these per-round values into a
secure-ranging-rate framework for QDB, analogous to a secure key rate in
quantum key distribution. Such a framework would combine \(\pDF\), \(\pMF\),
\(\pDFCV\), and \(\pMFCV\) with honest completeness, timing statistics,
multi-round thresholds, post-fast-phase checks, and the threat model to
quantify the rate at which secure ranging decisions can be certified.

\section*{Funding}

This work was supported by CyberSecurity Research Flanders with
reference number VR20192203, and by the European Union through the Horizon
Europe project \emph{Quantum Secure Networks Partnership} (QSNP, grant
agreement No. 101114043) and the Connecting Europe Facility project
BENELUX-QCI (grant agreement No. 101249520).

\section*{Code and data availability}
The code and data are available at
\url{https://github.com/kevinbogner/qdb-benchmark} and archived at
\url{https://doi.org/10.5281/zenodo.21361445}. The repository contains the
code that constructs the finite-dimensional DV-QDB SDP instances, solves the
primal and dual programs, and runs the restricted CV-QDB finite-grid
estimator, together with the scripts and parameters used in the paper. It
also contains explicit primal and dual witnesses and an exact-arithmetic
verifier for the four DV-QDB MF values. Saved floating-point solver outputs
are included as numerical reproducibility diagnostics.

\bibliographystyle{amsalpha}
\begingroup
\hbadness=10000
\bibliography{references}
\endgroup

\clearpage
\appendix
\crefalias{section}{appendix}
\crefname{appendix}{Appendix}{Appendices}
\Crefname{appendix}{Appendix}{Appendices}

\section{DV mafia-fraud certificate checks}
\label{app:mf-certificate-checks}

To check the reported MF values without trusting the SDP solver, we certify
each one with a matching pair of bounds. The primal SDP in \Cref{eq:mf-sdp}
maximizes the acceptance probability over all attack combs, so every feasible
comb \(S\) is an explicit attack, and its acceptance probability is a lower
bound on \(\pMF\). The upper bound needs the dual.

Write the averaged tester as
\[
\bar W:=\frac{1}{|\Omega|}\sum_{\omega\in\Omega} W_\omega,
\]
the mean of the per-branch acceptance operators over the branch set
\(\Omega\). With a Hermitian variable \(Y\) on \(Q\otimes R\otimes C\) and a
scalar \(\mu\), the dual is
\[
\min_{Y,\mu}\ \mu
\]
subject to
\[
Y\otimes I_O \succeq \bar W,
\qquad
\mu I_Q \succeq \operatorname{Tr}_{RC}(Y),
\qquad
Y=Y^\dagger,\quad \mu\in\mathbb R .
\]
The first constraint forces \(Y\) to dominate the averaged tester, and the
second caps \(\operatorname{Tr}_{RC}(Y)\) on the probe register \(Q\) by
\(\mu\). Together they give \(\operatorname{Tr}(\bar W S)\le\mu\) for every comb
\(S\). No attack then accepts with probability above \(\mu\), so any
dual-feasible \((Y,\mu)\) upper-bounds \(\pMF\).

A feasible primal and a feasible dual bound \(\pMF\) from below and above.
If their objective values are the same number \(v\), weak duality gives
\(v\leq\pMF\leq v\), and hence \(\pMF=v\). The artifact supplies such a
pair for every row of \Cref{tab:mf-exact-certificates}.

The exact verifier independently rebuilds \(\bar W\) from the finite branch
description of each protocol. It checks
\[
S\succeq0,\quad
\operatorname{Tr}_{O}(S)=\eta\otimes I_{RC},\quad
\eta\succeq0,\quad
\operatorname{Tr}(\eta)=1
\]
for the primal witness, both dual slack inequalities above, and the exact
identity \(\operatorname{Tr}(\bar W S)=\mu\). Matrix positivity is proved by
exact LDL decompositions. QDB 2017 and Mutual QDB use rational arithmetic,
QDB 2019 uses \(\mathbb Q(\sqrt5)\), and E91 QDB uses
\(\mathbb Q(\sqrt2,\sqrt{24+17\sqrt2})\). In the last field, algebraic signs
are decided with rational square-root enclosures. The verifier neither calls
an SDP solver nor reads the stored floating-point certificates.

\begin{table}[H]
\centering
\small
\caption{Exact DV MF values. For each row, the artifact verifies a feasible
primal and dual witness over the stated field with identical objective value.}
\label{tab:mf-exact-certificates}
\begin{tabular}{@{}lcc@{}}
\toprule
Protocol & Exact \(\pMF\) value & Certificate field\\
\midrule
QDB 2017 & \(1\) & \(\mathbb Q\)\\
QDB 2019 & \((5+\sqrt5)/8=\sin^{2}(2\pi/5)\) & \(\mathbb Q(\sqrt5)\)\\
Mutual QDB, first leg & \(3/4\) & \(\mathbb Q\)\\
E91 QDB retained & \(\bigl(8+\sqrt{24+17\sqrt2}\bigr)/16\) & \(\mathbb Q(\sqrt2,\sqrt{24+17\sqrt2})\)\\
\bottomrule
\end{tabular}
\end{table}

We also model all DV SDPs in CVXPY and solve them with SCS and CLARABEL.
A separate NumPy check rebuilds each tester and reports the residuals,
slack eigenvalues, and duality gap of the stored solver output. These
floating-point results are a reproducibility check, not the proof of exact
feasibility. They are shown in \Cref{tab:mf-numerical-checks}; the exact
certificate verifier is run with
\texttt{uv run qdb-benchmark exact-certificate {-}{-}protocol all}.

\begin{table}[H]
\centering
\small
\caption{NumPy diagnostics for the stored floating-point DV MF solver output.
The objective column is the stored primal objective. The final column is the
minimum over the two dual slack spectra. Its small negative values reflect
finite solver tolerance and show why this output is not itself a certificate;
the exact proof is summarized in \Cref{tab:mf-exact-certificates}.}
\label{tab:mf-numerical-checks}
\begin{tabular}{@{}lcccc@{}}
\toprule
Protocol & Objective & Duality gap & Comb residual & Minimum dual slack eigenvalue\\
\midrule
QDB 2017 & \(0.999999992\) & \(3.57\times10^{-9}\) & \(1.11\times10^{-16}\) & \(-3.26\times10^{-9}\)\\
QDB 2019 & \(0.904508487\) & \(3.37\times10^{-9}\) & \(1.82\times10^{-13}\) & \(-2.02\times10^{-9}\)\\
Mutual QDB, first leg & \(0.749999996\) & \(2.76\times10^{-9}\) & \(1.84\times10^{-16}\) & \(-1.43\times10^{-9}\)\\
E91 QDB retained & \(0.933200437\) & \(2.55\times10^{-11}\) & \(1.81\times10^{-11}\) & \(-2.46\times10^{-9}\)\\
\bottomrule
\end{tabular}
\end{table}

\section{QDB 2017 one-round values}
\label{app:qdb2017}

QDB 2017~\cite{abidin2017towards} derives the fast-phase basis string
from a PRF and authenticates the session with a final MAC over the
fast-phase transcript; the protocol flow is shown in
\Cref{fig:qdb2017-flow}. The fast phase confirms that the prover is
close (ranging), and the MAC that it is genuine (authentication). As
explained in \Cref{sec:applications}, the one-round fast-phase game
strips the MAC, so the main-text benchmark does not apply to this
protocol. We report its fast-phase-only values here for completeness.

In the notation of \Cref{sec:applications}, the reduced one-round
instantiation is as follows. The branch is \(\omega=(a,c)\), with DF
response-time view \(\view(\omega)=a\). The challenge state and accept
projector are \(\rho_\omega=\Pi_\omega=|c\rangle_a\langle c|\). The
honest-prover pre-ask channel is
\(\mathcal N_a(\rho)=\sum_t\langle t|_a\rho|t\rangle_a\,
|t\rangle_a\langle t|\) with \(t\) summed over \(\{0,1\}\): the same
measure-and-re-prepare map as QDB 2019, but with the response basis
equal to the challenge basis. The one-round values are \(\pDF=0.5\)
and \(\pMF=1\); the MF value is certified in
\Cref{app:mf-certificate-checks}.

Without the authentication the MAC provides, the adversary does not
even need the prover: it reflects the challenge state straight back to
the verifier, which accepts with probability \(1\). This reflection
attack is already noted in~\cite{abidin2017towards}, and the full
protocol uses the MAC to catch it. The value \(\pMF=1\) therefore
reflects the missing authentication, not a weakness of the full
protocol. It does show that ranging and authentication cannot be
separated in a DB protocol: a fast phase that only ranges provides no
MF resistance on its own.

\section{QDB protocol-flow diagrams}
\label{app:protocol-flows}

This appendix collects the protocol-flow diagrams for the QDB instantiations
in \Cref{sec:applications} and \Cref{app:qdb2017}.

\begin{figure}[tbp]
  \centering
  \resizebox{0.90\linewidth}{!}{%
  \begin{tikzpicture}[
      >=Latex,
      font=\small,
      actor/.style={draw, thick, fill=white, inner xsep=6pt, inner ysep=2pt, font=\bfseries},
      box/.style={draw, thick, fill=white, align=center, inner sep=4pt},
      smallbox/.style={draw, thick, fill=white, align=center, inner sep=3pt},
      timeline/.style={thick},
      msg/.style={thick, -Latex},
      bmsg/.style={thick, Latex-},
      qmsg/.style={thick, -Latex, decorate, decoration={snake, amplitude=0.35mm, segment length=2.2mm, pre length=0mm, post length=1.2mm}},
      qbmsg/.style={thick, Latex-, decorate, decoration={snake, amplitude=0.35mm, segment length=2.2mm, pre length=1.2mm, post length=0mm}}
  ]
    \def\LX{0}
    \def\RX{6.8}
    \pgfmathsetmacro{\MX}{0.5*(\LX+\RX)}
    \def\YEnd{-11.05}

    \node[actor] (V) at (\LX,0) {Verifier $(x)$};
    \node[actor] (P) at (\RX,0) {Prover $(x)$};

    \begin{scope}[yshift=-0.18cm]
    \node[box, minimum width=3.0cm] (Nv) at (\LX,-0.70) {$N_v \sample \{0,1\}^{\lambda}$};
    \node[box, minimum width=3.0cm] (Np) at (\RX,-0.70) {$N_p \sample \{0,1\}^{\lambda}$};

    \draw[timeline] (V.south) -- (Nv.north);
    \draw[timeline] (P.south) -- (Np.north);
    \draw[timeline] (\LX,-1.05) -- (\LX,\YEnd);
    \draw[timeline] (\RX,-1.05) -- (\RX,\YEnd);
    \draw[line width=1.6pt] (\LX-0.30,\YEnd) -- (\LX+0.30,\YEnd);
    \draw[line width=1.6pt] (\RX-0.30,\YEnd) -- (\RX+0.30,\YEnd);

    \draw[msg] (\LX,-1.50) -- (\RX,-1.50);
    \node at (\MX,-1.22) {$N_v$};
    \draw[bmsg] (\LX,-1.72) -- (\RX,-1.72);
    \node at (\MX,-2.00) {$N_p$};

    \node[box, minimum width=4.0cm] (Lc) at (\LX,-2.75)
      {$a := f_x(N_v,N_p) \in \{0,1\}^n$\\[-2pt]$c \overset{\$}{\leftarrow} \{0,1\}^n$};
    \node[box, minimum width=4.0cm] (Rc) at (\RX,-2.75)
      {$a := f_x(N_v,N_p) \in \{0,1\}^n$};

    \draw[dashed, thick] (-2.10,-3.68) rectangle (8.65,-7.40);

    \draw[qmsg] (\LX,-4.40) -- (\RX,-4.40);
    \node at (\MX,-4.12) {$|c_i\rangle_{a_i}$};

    \draw[qbmsg] (\LX,-6.30) -- (\RX,-6.30);
    \node at (\MX,-6.02) {$|c_{P,i}\rangle_{a_i}$};

    \draw[thick, decorate, decoration={brace, amplitude=5pt, mirror}]
      (-0.25,-4.40) -- (-0.25,-6.30)
      node[midway, left=4pt] {time $\Delta t_i$};

    \draw[thick, decorate, decoration={brace, amplitude=5pt}]
      (8.95,-3.73) -- (8.95,-7.40)
      node[midway, right=9pt, align=left] {Repeat\\for\\$i=1,\ldots,n$};

    \draw[bmsg] (\LX,-8.10) -- (\RX,-8.10);
    \node at (\MX,-7.82) {$\tau=\operatorname{MAC}_x(ID_{\verifier},ID_{\prover},N_v,N_p,c_P)$};

    \node[smallbox, minimum width=3.15cm] (measureBox) at (\RX,-5.35)
      {$c_{P,i} := \operatorname{Meas}_{a_i}(|c_i\rangle_{a_i})$};

    \node[smallbox, minimum width=3.65cm] (checkBox) at (\LX,-6.85)
      {$c_{V,i} := \operatorname{Meas}_{a_i}(|c_{P,i}\rangle_{a_i})$};

    \node[smallbox, minimum width=5.25cm, align=center] at (\LX,-9.60)
      {Accept iff\\
      $\tau=\operatorname{MAC}_x(ID_{\verifier},ID_{\prover},$\\[-1pt]
      $N_v,N_p,c_V)$,\\
      $\forall i:\ c_i=c_{V,i}$, and\\
      $\max_i \Delta t_i \le t_{\max}$};
    \end{scope}
  \end{tikzpicture}%
  }
  \caption{Overview of the QDB 2017 protocol~\cite{abidin2017towards}. Before the fast rounds, the parties exchange nonces and derive the PRF output \(a\). In each fast round, the verifier prepares \(|c_i\rangle_{a_i}\), and the prover measures it in basis \(a_i\) and re-prepares the outcome in the same basis. After all rounds, the prover sends \(\tau=\operatorname{MAC}_x(ID_{\verifier},ID_{\prover},N_v,N_p,c_P)\). The verifier accepts only if \(\tau\) equals the MAC recomputed on \(c_V\), all recovered bits match the originals, and every round-trip time \(\Delta t_i\) stays below \(t_{\max}\).}
  \label{fig:qdb2017-flow}
\end{figure}
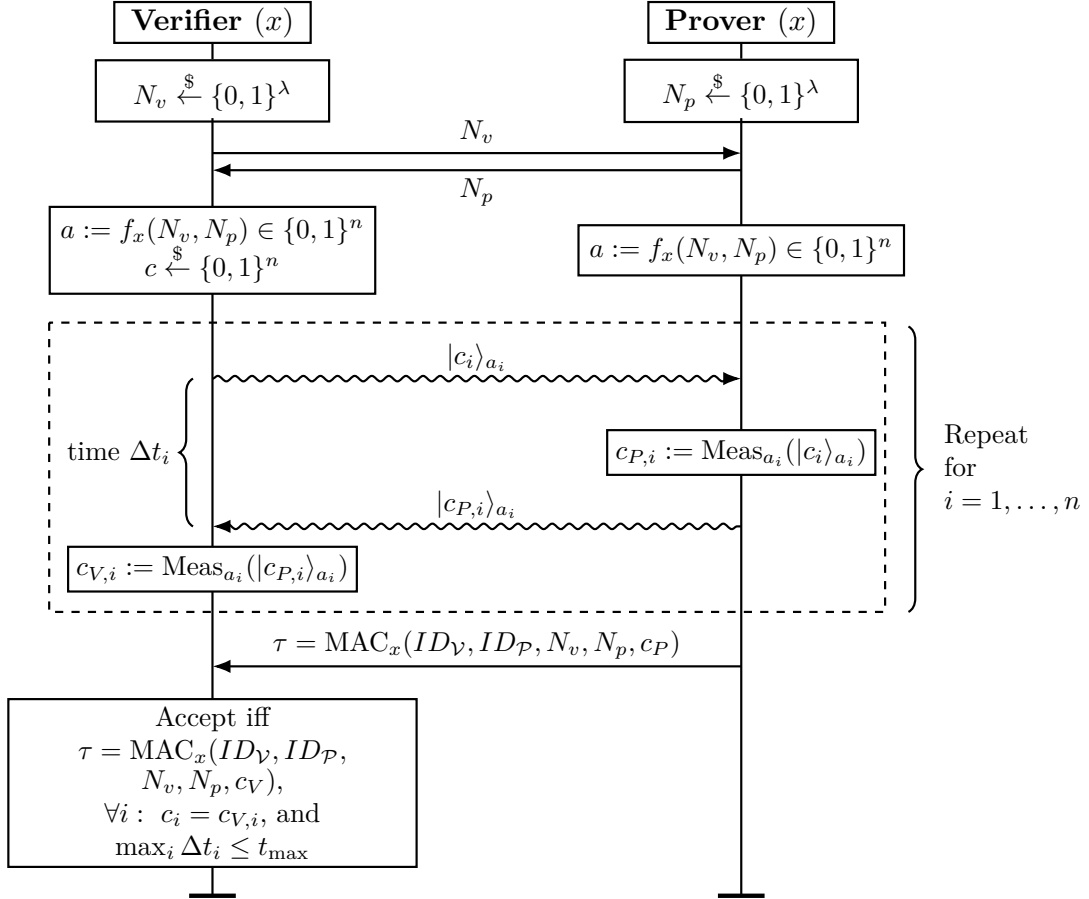

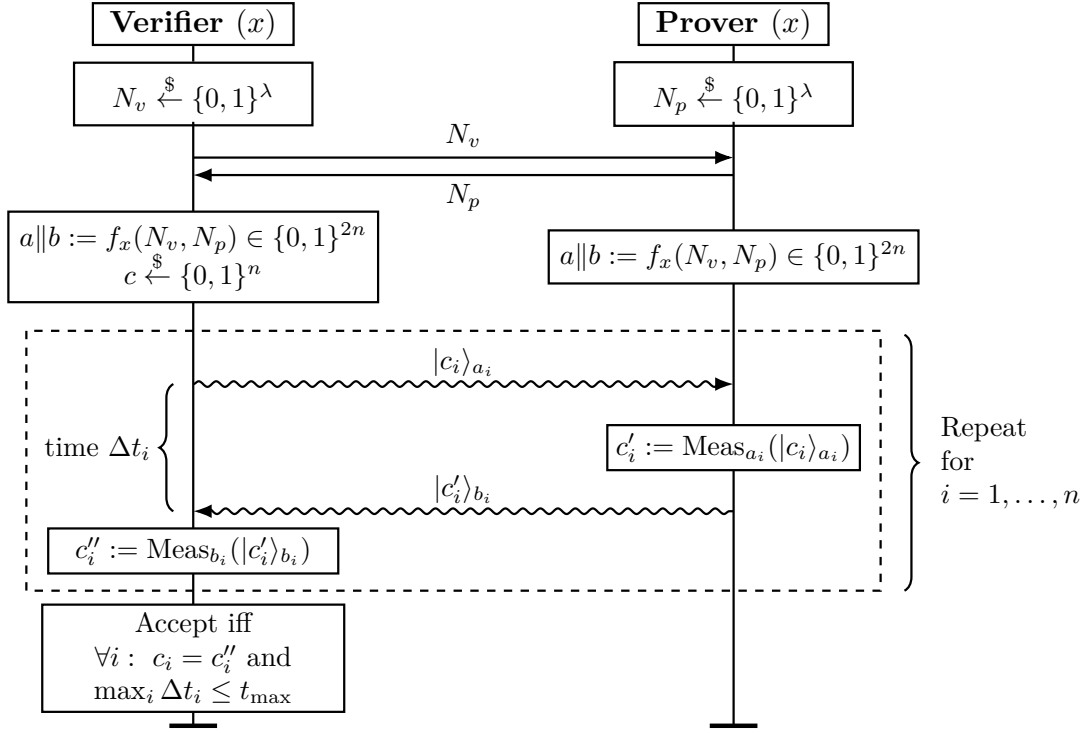
\begin{figure}[H]
  \centering
  \resizebox{0.90\linewidth}{!}{%
  \begin{tikzpicture}[
      >=Latex,
      font=\small,
      actor/.style={draw, thick, fill=white, inner xsep=6pt, inner ysep=2pt, font=\bfseries},
      box/.style={draw, thick, fill=white, align=center, inner sep=4pt},
      smallbox/.style={draw, thick, fill=white, align=center, inner sep=3pt},
      timeline/.style={thick},
      msg/.style={thick, -Latex},
      bmsg/.style={thick, Latex-},
      qmsg/.style={thick, -Latex, decorate, decoration={snake, amplitude=0.35mm, segment length=2.2mm, pre length=0mm, post length=1.2mm}},
      qbmsg/.style={thick, Latex-, decorate, decoration={snake, amplitude=0.35mm, segment length=2.2mm, pre length=1.2mm, post length=0mm}}
  ]
    \def\LX{0}
    \def\RX{6.8}
    \pgfmathsetmacro{\MX}{0.5*(\LX+\RX)}
    \def\YEnd{-8.65}

    \node[actor] (V) at (\LX,0) {Verifier $(x)$};
    \node[actor] (P) at (\RX,0) {Prover $(x)$};

    \begin{scope}[yshift=-0.18cm]
    \node[box, minimum width=3.0cm] (Nv) at (\LX,-0.70) {$N_v \sample \{0,1\}^{\lambda}$};
    \node[box, minimum width=3.0cm] (Np) at (\RX,-0.70) {$N_p \sample \{0,1\}^{\lambda}$};

    \draw[timeline] (V.south) -- (Nv.north);
    \draw[timeline] (P.south) -- (Np.north);
    \draw[timeline] (\LX,-1.05) -- (\LX,\YEnd);
    \draw[timeline] (\RX,-1.05) -- (\RX,\YEnd);
    \draw[line width=1.6pt] (\LX-0.30,\YEnd) -- (\LX+0.30,\YEnd);
    \draw[line width=1.6pt] (\RX-0.30,\YEnd) -- (\RX+0.30,\YEnd);

    \draw[msg] (\LX,-1.50) -- (\RX,-1.50);
    \node at (\MX,-1.22) {$N_v$};
    \draw[bmsg] (\LX,-1.72) -- (\RX,-1.72);
    \node at (\MX,-2.00) {$N_p$};

    \node[box, minimum width=4.25cm] (Lc) at (\LX,-2.75)
      {$a\|b := f_x(N_v,N_p) \in \{0,1\}^{2n}$\\[-2pt]$c \overset{\$}{\leftarrow} \{0,1\}^n$};
    \node[box, minimum width=4.25cm] (Rc) at (\RX,-2.75)
      {$a\|b := f_x(N_v,N_p) \in \{0,1\}^{2n}$};

    \draw[dashed, thick] (-2.10,-3.68) rectangle (8.65,-6.95);

    \draw[qmsg] (\LX,-4.35) -- (\RX,-4.35);
    \node at (\MX,-4.07) {$|c_i\rangle_{a_i}$};

    \draw[qbmsg] (\LX,-5.95) -- (\RX,-5.95);
    \node at (\MX,-5.67) {$|c_i'\rangle_{b_i}$};

    \draw[thick, decorate, decoration={brace, amplitude=5pt, mirror}]
      (-0.25,-4.35) -- (-0.25,-5.95)
      node[midway, left=4pt] {time $\Delta t_i$};

    \draw[thick, decorate, decoration={brace, amplitude=5pt}]
      (8.95,-3.73) -- (8.95,-6.95)
      node[midway, right=9pt, align=left] {Repeat\\for\\$i=1,\ldots,n$};

    \node[smallbox, minimum width=3.15cm] (measureBox) at (\RX,-5.15)
      {$c_i' := \operatorname{Meas}_{a_i}(|c_i\rangle_{a_i})$};

    \node[smallbox, minimum width=3.65cm] (checkBox) at (\LX,-6.45)
      {$c_i'' := \operatorname{Meas}_{b_i}(|c_i'\rangle_{b_i})$};

    \node[smallbox, minimum width=3.8cm, align=center] at (\LX,-7.80)
      {Accept iff\\
      $\forall i:\ c_i=c_i''$ and\\
      $\max_i \Delta t_i \le t_{\max}$};
    \end{scope}
  \end{tikzpicture}%
  }
  \caption{Overview of the QDB 2019 protocol~\cite{abidin2019quantum}. Before the fast rounds, the parties exchange nonces and derive the PRF outputs \(a\) and \(b\). In each fast round, the verifier prepares \(|c_i\rangle_{a_i}\), and the prover measures it in basis \(a_i\) and re-prepares the outcome in basis \(b_i\). After all rounds, the verifier accepts only if all recovered bits match the originals and every round-trip time \(\Delta t_i\) stays below \(t_{\max}\).}
  \label{fig:qdb2019-flow}
\end{figure}

\begin{figure}[H]
  \centering
  \resizebox{0.90\linewidth}{!}{%
  \begin{tikzpicture}[
      >=Latex,
      font=\small,
      actor/.style={draw, thick, fill=white, inner xsep=6pt, inner ysep=2pt, font=\bfseries},
      box/.style={draw, thick, fill=white, align=center, inner sep=4pt},
      smallbox/.style={draw, thick, fill=white, align=center, inner sep=3pt},
      timeline/.style={thick},
      msg/.style={thick, -Latex},
      bmsg/.style={thick, Latex-},
      qmsg/.style={thick, -Latex, decorate, decoration={snake, amplitude=0.35mm, segment length=2.2mm, pre length=0mm, post length=1.2mm}},
      qbmsg/.style={thick, Latex-, decorate, decoration={snake, amplitude=0.35mm, segment length=2.2mm, pre length=1.2mm, post length=0mm}}
  ]
    \def\LX{0}
    \def\RX{7.4}
    \pgfmathsetmacro{\MX}{0.5*(\LX+\RX)}
    \def\YEnd{-14.80}

    \node[actor] (V) at (\LX,0) {Verifier $(x)$};
    \node[actor] (P) at (\RX,0) {Prover $(x)$};

    \begin{scope}[yshift=-0.35cm]
    \draw[timeline] (V.south) -- (\LX,\YEnd);
    \draw[timeline] (P.south) -- (\RX,\YEnd);
    \draw[line width=1.6pt] (\LX-0.30,\YEnd) -- (\LX+0.30,\YEnd);
    \draw[line width=1.6pt] (\RX-0.30,\YEnd) -- (\RX+0.30,\YEnd);

    \node[box, minimum width=3.35cm] at (\LX,-0.75)
      {$N_v \sample \{0,1\}^{\lambda}$};
    \node[box, minimum width=3.35cm] at (\RX,-0.75)
      {$N_p \sample \{0,1\}^{\lambda}$};

    \draw[msg] (\LX,-1.55) -- (\RX,-1.55);
    \node at (\MX,-1.27) {$N_v$};
    \draw[bmsg] (\LX,-1.77) -- (\RX,-1.77);
    \node at (\MX,-2.05) {$N_p$};

    \node[box, minimum width=4.75cm] at (\LX,-2.65)
      {$a\|b\|\gamma := f_x(N_v,N_p) \in \{0,1\}^{3n}$};
    \node[box, minimum width=4.75cm] at (\RX,-2.65)
      {$a\|b\|\gamma := f_x(N_v,N_p) \in \{0,1\}^{3n}$\\[2pt]$r \sample \{0,1\}^n$};

    \draw[bmsg] (\LX,-3.65) -- (\RX,-3.65);
    \node at (\MX,-3.37) {$\operatorname{Com}(r)$};

    \draw[dashed, thick] (-2.85,-4.20) rectangle (10.05,-11.40);

    \node[smallbox, minimum width=4.80cm] at (\LX,-5.00)
      {Prepare EPR pair $(EP_{v,i},EP_{p,i})$\\[1pt]$m_i := \operatorname{Meas}_{a_i}(EP_{v,i})$};

    \draw[qmsg] (\LX,-5.95) -- (\RX,-5.95);
    \node at (\MX,-5.67) {$EP_{p,i}$};

    \node[smallbox, minimum width=4.80cm] at (\RX,-6.95)
      {$m_i' := \operatorname{Meas}_{a_i}(EP_{p,i})$\\[1pt]prepare $|m_i' \oplus r_i\rangle_{b_i}$};

    \draw[qbmsg] (\LX,-7.90) -- (\RX,-7.90);
    \node at (\MX,-7.62) {$|m_i' \oplus r_i\rangle_{b_i}$};

    \node[smallbox, minimum width=4.95cm] at (\LX,-8.90)
      {$d_i := \operatorname{Meas}_{b_i}(|m_i' \oplus r_i\rangle_{b_i})$\\[1pt]recover $r_i' := d_i \oplus m_i$};

    \draw[qmsg] (\LX,-9.85) -- (\RX,-9.85);
    \node at (\MX,-9.57) {$|r_i'\rangle_{\gamma_i}$};

    \node[smallbox, minimum width=4.60cm] at (\RX,-10.70)
      {$r_i'' := \operatorname{Meas}_{\gamma_i}(|r_i'\rangle_{\gamma_i})$\\[1pt]check $r_i''=r_i$};

    \draw[thick, decorate, decoration={brace, amplitude=5pt, mirror}]
      (-0.25,-5.95) -- (-0.25,-7.90)
      node[midway, left=4pt] {time $\Delta t_i^{(1)}$};

    \draw[thick, decorate, decoration={brace, amplitude=5pt}]
      (7.65,-7.90) -- (7.65,-9.85)
      node[midway, right=5pt] {time $\Delta t_i^{(2)}$};

    \draw[thick, decorate, decoration={brace, amplitude=5pt}]
      (10.35,-4.25) -- (10.35,-11.40)
      node[midway, right=9pt, align=left] {Repeat\\for\\$i=1,\ldots,n$};

    \draw[bmsg] (\LX,-12.15) -- (\RX,-12.15);
    \node at (\MX,-11.87) {$m', \operatorname{Open}(r)$};

    \node[smallbox, minimum width=8.10cm, align=center] at (\MX,-13.55)
      {Accept iff\\
      the commitment opening is valid,\\
      $\forall i:\ r_i''=r_i$, $m' = m_1\|\cdots\|m_n$, and\\
      $\max_i \Delta t_i^{(1)} \le t_{\max}$ and $\max_i \Delta t_i^{(2)} \le t_{\max}$};
    \end{scope}
  \end{tikzpicture}%
  }
  \caption{Overview of the Mutual QDB protocol~\cite{abidin2025entanglement}. Before the fast rounds, the parties exchange nonces and derive the PRF outputs \(a\), \(b\), and \(\gamma\), and the prover commits to a random mask string \(r\). In each fast round, the verifier prepares an EPR pair, measures its own half in basis \(a_i\), and sends the prover half \(EP_{p,i}\). The prover measures in the same basis and returns \(|m_i' \oplus r_i\rangle_{b_i}\). The verifier then recovers the mask bit and sends back \(|r_i'\rangle_{\gamma_i}\). After all rounds, the prover opens the commitment and reveals \(m'\). The verifier accepts only if the commitment opening is valid, the revealed measurement string is consistent with its recorded outcomes, and every first-leg round-trip time stays below \(t_{\max}\). The prover accepts only if its committed mask bits come back unchanged and every second-leg round-trip time stays below \(t_{\max}\).}
\label{fig:mutualqdb-flow}
\end{figure}
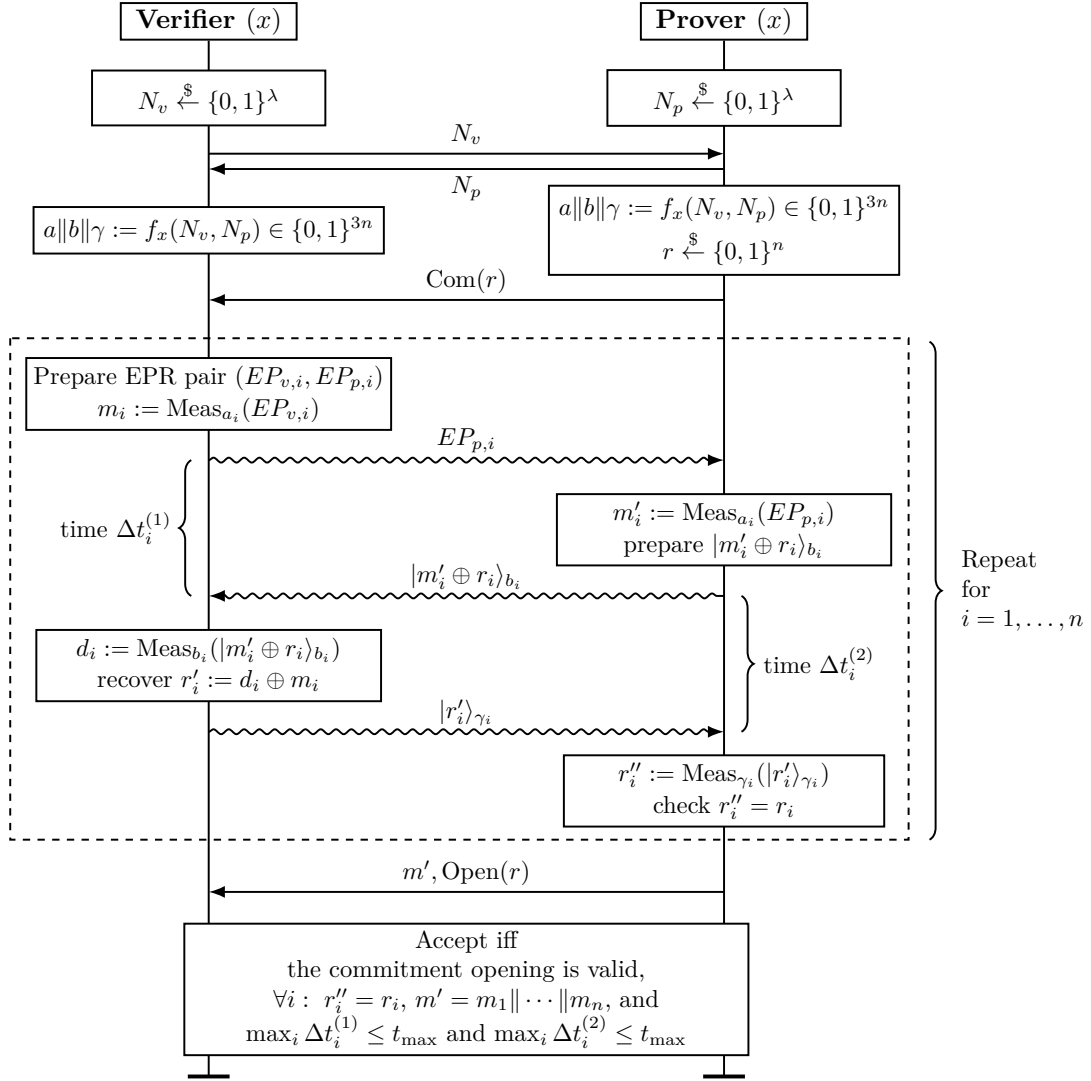

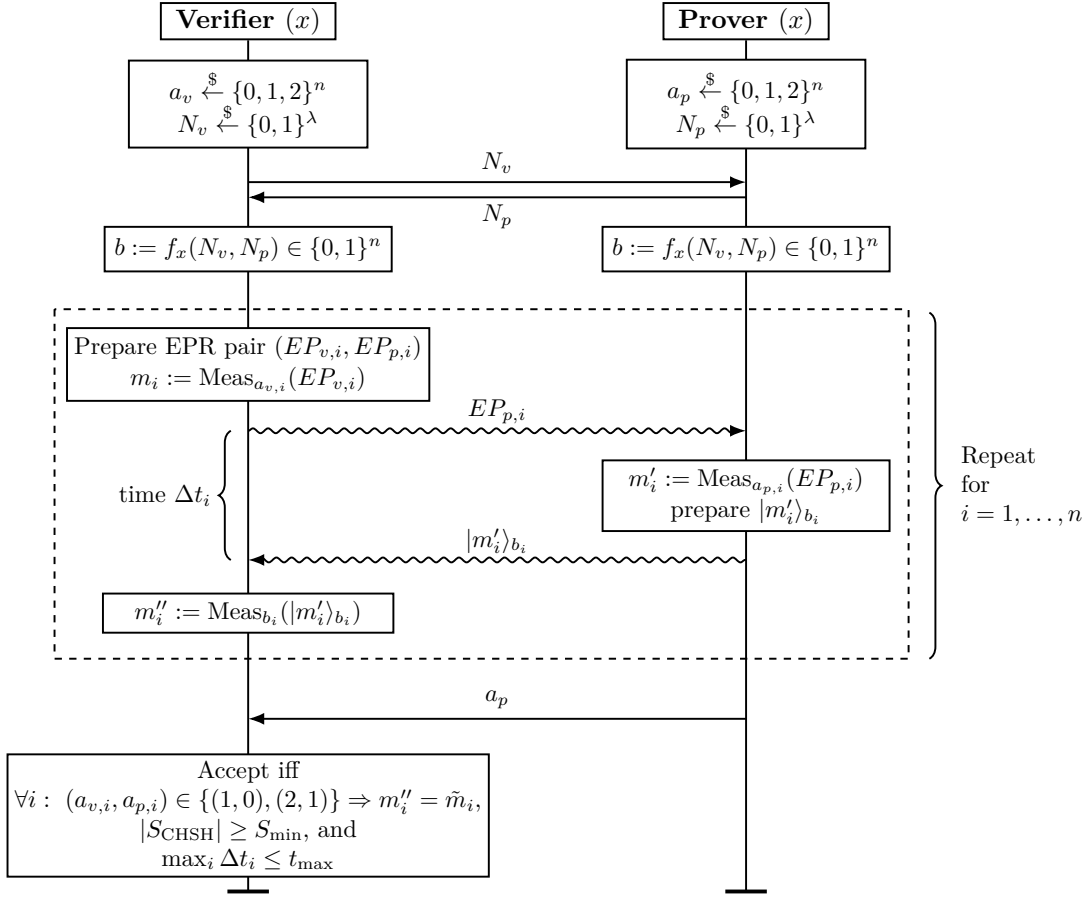
\begin{figure}[H]
  \centering
  \resizebox{0.90\linewidth}{!}{%
  \begin{tikzpicture}[
      >=Latex,
      font=\small,
      actor/.style={draw, thick, fill=white, inner xsep=6pt, inner ysep=2pt, font=\bfseries},
      box/.style={draw, thick, fill=white, align=center, inner sep=4pt},
      smallbox/.style={draw, thick, fill=white, align=center, inner sep=3pt},
      timeline/.style={thick},
      msg/.style={thick, -Latex},
      bmsg/.style={thick, Latex-},
      qmsg/.style={thick, -Latex, decorate, decoration={snake, amplitude=0.35mm, segment length=2.2mm, pre length=0mm, post length=1.2mm}},
      qbmsg/.style={thick, Latex-, decorate, decoration={snake, amplitude=0.35mm, segment length=2.2mm, pre length=1.2mm, post length=0mm}}
  ]
    \def\LX{0}
    \def\RX{7.2}
    \pgfmathsetmacro{\MX}{0.5*(\LX+\RX)}
    \def\YEnd{-12.15}

    \node[actor] (V) at (\LX,0) {Verifier $(x)$};
    \node[actor] (P) at (\RX,0) {Prover $(x)$};

    \begin{scope}[yshift=-0.45cm]
    \draw[timeline] (V.south) -- (\LX,\YEnd);
    \draw[timeline] (P.south) -- (\RX,\YEnd);
    \draw[line width=1.6pt] (\LX-0.30,\YEnd) -- (\LX+0.30,\YEnd);
    \draw[line width=1.6pt] (\RX-0.30,\YEnd) -- (\RX+0.30,\YEnd);

    \node[box, minimum width=3.45cm] (Lv) at (\LX,-0.75)
      {$a_v \sample \{0,1,2\}^n$\\[-2pt]$N_v \sample \{0,1\}^{\lambda}$};
    \node[box, minimum width=3.45cm] (Rp) at (\RX,-0.75)
      {$a_p \sample \{0,1,2\}^n$\\[-2pt]$N_p \sample \{0,1\}^{\lambda}$};

    \draw[msg] (\LX,-1.88) -- (\RX,-1.88);
    \node at (\MX,-1.60) {$N_v$};
    \draw[bmsg] (\LX,-2.10) -- (\RX,-2.10);
    \node at (\MX,-2.38) {$N_p$};

    \node[box, minimum width=4.0cm] at (\LX,-2.85)
      {$b := f_x(N_v,N_p) \in \{0,1\}^n$};
    \node[box, minimum width=4.0cm] at (\RX,-2.85)
      {$b := f_x(N_v,N_p) \in \{0,1\}^n$};

    \draw[dashed, thick] (-2.80,-3.72) rectangle (9.55,-8.78);

    \node[smallbox, minimum width=4.55cm] at (\LX,-4.52)
      {Prepare EPR pair $(EP_{v,i},EP_{p,i})$\\[1pt]$m_i := \operatorname{Meas}_{a_{v,i}}(EP_{v,i})$};

    \draw[qmsg] (\LX,-5.48) -- (\RX,-5.48);
    \node at (\MX,-5.20) {$EP_{p,i}$};

    \node[smallbox, minimum width=4.15cm] at (\RX,-6.42)
      {$m_i' := \operatorname{Meas}_{a_{p,i}}(EP_{p,i})$\\[1pt]prepare $|m_i'\rangle_{b_i}$};

    \draw[qbmsg] (\LX,-7.35) -- (\RX,-7.35);
    \node at (\MX,-7.07) {$|m_i'\rangle_{b_i}$};

    \node[smallbox, minimum width=4.20cm] at (\LX,-8.12)
      {$m_i'' := \operatorname{Meas}_{b_i}(|m_i'\rangle_{b_i})$};

    \draw[thick, decorate, decoration={brace, amplitude=5pt, mirror}]
      (-0.25,-5.48) -- (-0.25,-7.35)
      node[midway, left=4pt] {time $\Delta t_i$};

    \draw[thick, decorate, decoration={brace, amplitude=5pt}]
      (9.85,-3.77) -- (9.85,-8.78)
      node[midway, right=9pt, align=left] {Repeat\\for\\$i=1,\ldots,n$};

    \draw[bmsg] (\LX,-9.65) -- (\RX,-9.65);
    \node at (\MX,-9.37) {$a_p$};

    \node[smallbox, minimum width=6.95cm, align=center] at (\LX,-11.05)
      {Accept iff\\
      $\forall i:\ (a_{v,i},a_{p,i})\in\{(1,0),(2,1)\}\Rightarrow m_i''=\tilde m_i$,\\
      $|S_{\mathrm{CHSH}}| \ge S_{\min}$, and\\
      $\max_i \Delta t_i \le t_{\max}$};
    \end{scope}
  \end{tikzpicture}%
  }
  \caption{Overview of the E91 QDB protocol~\cite{bogner2024entangled}. Before the fast rounds, the parties choose the measurement-setting strings \(a_v\) and \(a_p\), exchange nonces, and derive the response-basis string \(b\). In each fast round, the verifier prepares an entangled pair, measures its own half in setting \(a_{v,i}\), and sends the prover half \(EP_{p,i}\). The prover measures in setting \(a_{p,i}\) and returns \(|m_i'\rangle_{b_i}\). After all rounds, the prover reveals \(a_p\). The verifier accepts only if the logical target \(\tilde m_i\) matches on the retained setting pairs \((1,0)\) and \((2,1)\), the sign-adjusted CHSH quantity satisfies \(|S_{\mathrm{CHSH}}|\ge S_{\min}\) on \((0,0)\), \((0,2)\), \((2,0)\), and \((2,2)\), and every round-trip time \(\Delta t_i\) stays below \(t_{\max}\). The two conventions for the logical target \(\tilde m_i\) are stated in \Cref{sec:applications} (E91 QDB).}
  \label{fig:e91qdb-flow}
\end{figure}

\begin{figure}[H]
  \centering
  \resizebox{0.90\linewidth}{!}{%
  \begin{tikzpicture}[
      >=Latex,
      font=\small,
      actor/.style={draw, thick, fill=white, inner xsep=6pt, inner ysep=2pt, font=\bfseries},
      box/.style={draw, thick, fill=white, align=center, inner sep=4pt},
      smallbox/.style={draw, thick, fill=white, align=center, inner sep=3pt},
      timeline/.style={thick},
      msg/.style={thick, -Latex},
      bmsg/.style={thick, Latex-},
      qmsg/.style={thick, -Latex, decorate, decoration={snake, amplitude=0.35mm, segment length=2.2mm, pre length=0mm, post length=1.2mm}},
      qbmsg/.style={thick, Latex-, decorate, decoration={snake, amplitude=0.35mm, segment length=2.2mm, pre length=1.2mm, post length=0mm}}
  ]
    \def\LX{0}
    \def\RX{6.9}
    \pgfmathsetmacro{\MX}{0.5*(\LX+\RX)}
    \def\YEnd{-9.65}

    \node[actor] (V) at (\LX,0) {Verifier $(x)$};
    \node[actor] (P) at (\RX,0) {Prover $(x)$};

    \begin{scope}[yshift=-0.18cm]
    \node[box, minimum width=3.0cm] (Nv) at (\LX,-0.70) {$N_v \sample \{0,1\}^{\lambda}$};
    \node[box, minimum width=3.0cm] (Np) at (\RX,-0.70) {$N_p \sample \{0,1\}^{\lambda}$};

    \draw[timeline] (V.south) -- (Nv.north);
    \draw[timeline] (P.south) -- (Np.north);
    \draw[timeline] (\LX,-1.05) -- (\LX,\YEnd);
    \draw[timeline] (\RX,-1.05) -- (\RX,\YEnd);
    \draw[line width=1.6pt] (\LX-0.30,\YEnd) -- (\LX+0.30,\YEnd);
    \draw[line width=1.6pt] (\RX-0.30,\YEnd) -- (\RX+0.30,\YEnd);

    \draw[msg] (\LX,-1.50) -- (\RX,-1.50);
    \node at (\MX,-1.22) {$N_v$};
    \draw[bmsg] (\LX,-1.72) -- (\RX,-1.72);
    \node at (\MX,-2.00) {$N_p$};

    \begin{scope}[yshift=0.15cm]
    \node[box, minimum width=4.25cm] (Lc) at (\LX,-2.75)
      {$a\|b := f_x(N_v,N_p) \in \{0,1\}^{2n}$};
    \node[box, minimum width=4.25cm] (Rc) at (\RX,-2.75)
      {$a\|b := f_x(N_v,N_p) \in \{0,1\}^{2n}$};

    \begin{scope}[yshift=0.40cm]
    \draw[dashed, thick] (-2.35,-3.81) rectangle (8.75,-8.51);

    \node[smallbox, minimum width=4.2cm] at (\LX,-4.45)
      {Prepare EPR pair $(A_i,B_i)$\\[-2pt]$m_i := \mathrm{HD}_{a_i}(A_i)$};

    \draw[qmsg] (\LX,-5.35) -- (\RX,-5.35);
    \node at (\MX,-5.07) {EPR mode $B_i$};

    \node[smallbox, minimum width=3.10cm] at (\RX,-6.20)
      {$m_i' := \mathrm{HD}_{a_i}(B_i)$};

    \draw[qbmsg] (\LX,-7.20) -- (\RX,-7.20);
    \node at (\MX,-6.92) {Prepare mode $C_i$ from $m_i'$ using $b_i$};

    \node[smallbox, minimum width=3.95cm] at (\LX,-8.05)
      {$m_i'' := \mathrm{HD}_{b_i}(C_i)$};

    \draw[thick, decorate, decoration={brace, amplitude=5pt, mirror}]
      (-0.25,-5.35) -- (-0.25,-7.20)
      node[midway, left=4pt] {time $\Delta t_i$};

    \draw[thick, decorate, decoration={brace, amplitude=5pt}]
      (9.05,-3.86) -- (9.05,-8.51)
      node[midway, right=9pt, align=left] {Repeat\\for\\$i=1,\ldots,n$};
    \end{scope}

    \node[smallbox, minimum width=4.2cm, align=center] at (\LX,-8.95)
      {Accept iff\\
      $\frac{1}{n}\sum_i \mathbf 1\{|m_i-m_i''|\le\delta_{\mathrm{cv}}\}
      \ge \tau_{\mathrm{cv}}$ and\\
      $\mathsf T(\Delta t_1,\ldots,\Delta t_n)\le \tau_t$};
    \end{scope}
    \end{scope}
  \end{tikzpicture}%
  }
  \caption{Overview of the Gaussian CV-QDB protocol~\cite{bogner2025continuous}. Before the fast rounds, the parties exchange nonces and derive the PRF outputs \(a\) and \(b\). In each fast round, the verifier prepares an EPR pair, measures one mode in basis \(a_i\), and sends the other mode to the prover. The prover measures that mode in basis \(a_i\) and prepares a response mode from the outcome using \(b_i\). After all rounds, the verifier accepts only if the returned quadrature lies within tolerance \(\delta_{\mathrm{cv}}\) in a sufficient fraction of rounds and the round-trip times pass the timing check \(\mathsf T(\Delta t_1,\ldots,\Delta t_n)\le\tau_t\).}
  \label{fig:cvqdb-flow}
\end{figure}
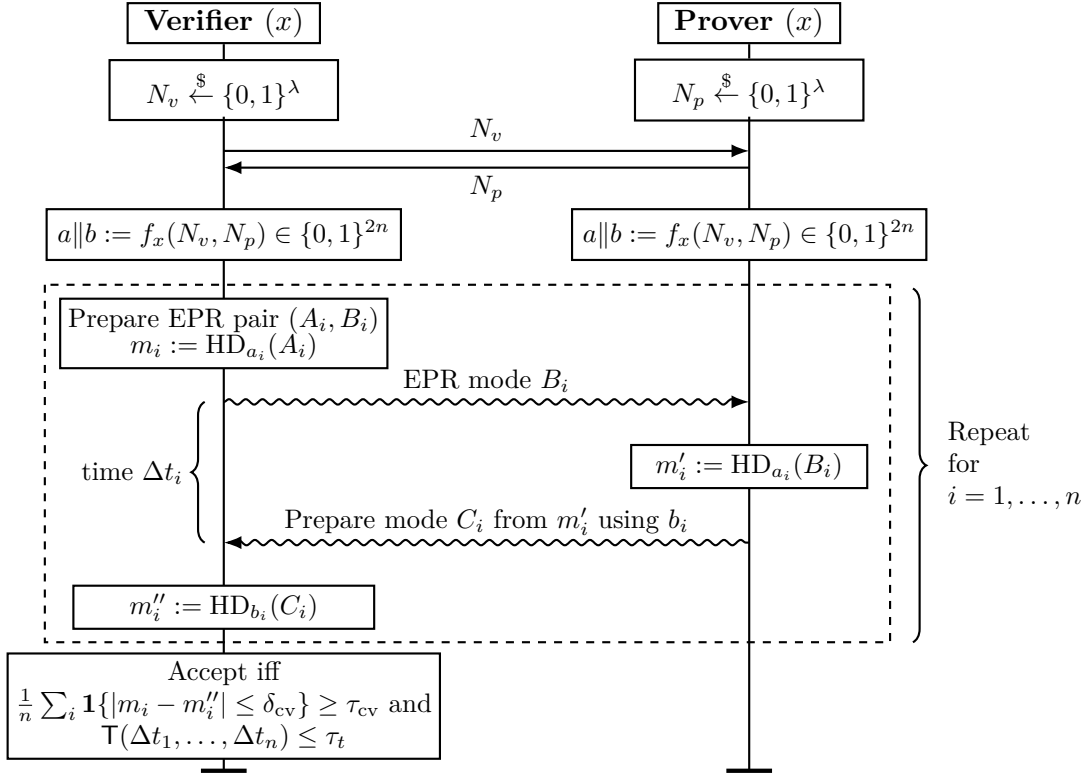

\section{CV-QDB finite-grid Gaussian MF estimator details}
\label{app:cv-grid-benchmark}

This appendix gives the grid and estimator details behind the restricted
Gaussian MF estimator of \Cref{sec:cvqdb}, computed with zero added response
noise. A Gaussian attack depends on only a few continuous parameters, such as
the probe amplitude and the measurement angles. We restrict each one to a
finite set of values and optimize over the resulting grid. Refining the grid
may expose stronger attacks.

At a high level, the attack has three steps. First, the adversary sends a
Gaussian probe to the honest prover and measures the returned mode, obtaining a
pre-ask value \(Z_\phi\) that leaks information about the hidden bases. Second,
when the verifier's challenge mode arrives, the adversary picks a measurement
angle \(\theta\) from \(Z_\phi\) and reads off the outcome \(Y_\theta\). Third,
it turns \(Y_\theta\) into a linear estimate of the verifier's reference
quadrature and returns that as its response. The grid searches over the
parameters of the first two steps for the estimate that passes most often.

Concretely, the adversary sends a Gaussian probe mode \(Q\). The prover
measures the quadrature selected by \(a\) and returns a response mode whose
\(b\)-quadrature contains \(g_a U_a\),
where \(U_a=X_Q\) for \(a=0\), \(U_a=P_Q\) for \(a=1\), and
\(g_0=1,g_1=-1\). The adversary measures this returned mode at angle \(\phi\),
obtaining \(Z_\phi\). Later, after receiving the verifier's actual challenge
mode, it chooses a challenge-measurement angle \(\theta(Z_\phi)\), obtains
\(Y_{\theta(Z_\phi)}\), and prepares a response mode with quadrature estimates
\(\widehat M_0\) and \(\widehat M_1\). We write \(M_a\) for the verifier's
reference quadrature in basis \(a\). The response estimates are scored with
zero added noise in the measured response quadrature (an attacker-favorable
relaxation, not a certified finite-energy output-mode model).

Ideally, the adversary would maximize its average chance of landing within
\(\delta_{\mathrm{cv}}\) of the verifier's reference quadrature,
\[
\frac14
\sum_{a,b}
\Pr\left[
|M_a-\widehat M_b(Z_\phi,Y_{\theta(Z_\phi)})|
\le
\delta_{\mathrm{cv}}
\right],
\]
over the energy-constrained Gaussian probe, the pre-ask measurement, the
adaptive challenge measurement, and the Gaussian response estimators. To our
knowledge, no technique solves this continuous nonconvex optimization exactly,
so we report the finite-grid estimate \(\pMFCV\) instead, with the grid and
energy bound below.

The reported grid fixes the pre-ask probe to an unsqueezed displaced Gaussian
state. Its displacement angle is \(\pi/4\) and its amplitude runs from \(0\) to
\(3.0\) in steps of \(0.1\), subject to \(N_{\rm pre}\le2\). The pre-ask angle
\(\phi\) and the challenge angle \(\theta\) each range over \(181\) points
uniformly spaced in \([0, 2\pi)\), and the pre-ask outcome \(Z_\phi\) over a
\(1001\)-point quadrature grid. The measured response quadrature carries zero
added noise in each response basis \(b\), matching the noiseless DV upper
bound, while the inactive quadrature is assigned vacuum variance
\(\hbar/2=1\). The pre-ask energy uses the \(\hbar=2\) convention,
\[
N_{\rm pre}
=
\frac{
\operatorname{Var}(X_Q)+\operatorname{Var}(P_Q)
+\mathbb E[X_Q]^2+\mathbb E[P_Q]^2
}{2\hbar}
-\frac12 .
\]

The \(Z_\phi\)-grid is uniform on
\[
\left[
\min_{a,b}(\mu_{ab}-8\sqrt{v_{ab}}),
\max_{a,b}(\mu_{ab}+8\sqrt{v_{ab}})
\right],
\]
where \(\mu_{ab}\) and \(v_{ab}\) are the mean and variance of \(Z_\phi\) for
the branch labeled by \((a,b)\), and the integral over \(Z_\phi\) is evaluated
by the trapezoidal rule. For each pre-ask outcome \(Z_\phi=z\), let
\(q_{ab}(z)=\Pr[a,b\mid Z_\phi=z]\) be the posterior probability that the
hidden bases are \((a,b)\) given \(z\). For a candidate challenge angle \(\theta\), let \(Y_\theta\) be the measured
challenge quadrature. The adversary does not know the hidden bases, so it
estimates \(M_a\) from \(Y_\theta\) by linear regression and weights the
regression coefficient by the posterior over \(a\),
\[
\widehat M_b=\kappa_b(z,\theta)Y_\theta,\qquad
\kappa_b(z,\theta)=
\frac{
\sum_a q_{ab}(z)\operatorname{Cov}(M_a,Y_\theta)/\operatorname{Var}(Y_\theta)
}{
\sum_a q_{ab}(z)
},
\]
with \(\kappa_b=0\) when the denominator is zero. The conditional error
variance for branch \((a,b)\), written \(V_{\mathrm{err}}\), is
\[
\operatorname{Var}(M_a)+\kappa_b^2\operatorname{Var}(Y_\theta)
-2\kappa_b\operatorname{Cov}(M_a,Y_\theta)
+\operatorname{Var}_{\mathrm{resp}}(b).
\]
The branch then passes with probability
\[
\Pr(|\Delta|\le \delta_{\mathrm{cv}})
=
\operatorname{erf}\left(
\frac{\delta_{\mathrm{cv}}}{\sqrt{2V_{\mathrm{err}}}}
\right).
\]
A larger \(\operatorname{Var}_{\mathrm{resp}}(b)\) therefore raises
\(V_{\mathrm{err}}\) and lowers this probability for a fixed estimator, so
setting the response noise to zero is attacker-favorable within this estimator
model. It is not, however, a physical finite-energy constraint. If the
adversary must output a single CV mode while the hidden basis \(b\) is
uncertain, its two response quadratures jointly obey a covariance uncertainty
bound.
We do not optimize over nonlinear estimators or over physical two-quadrature
response covariance constraints. With this restricted grid, we obtain the
zero-response-noise estimator
\[
\pMFCV(N_{\rm pre}\le 2)\approx 0.913815.
\]

\end{document}